\documentclass[10pt,conference]{IEEEtran}

\renewcommand\thesection{\arabic{section}}
\renewcommand\thesubsection{\thesection.\arabic{subsection}}
\renewcommand\thesubsubsection{\thesubsection.\arabic{subsubsection}}

\renewcommand\thesectiondis{\arabic{section}}
\renewcommand\thesubsectiondis{\thesectiondis.\arabic{subsection}}
\renewcommand\thesubsubsectiondis{\thesubsectiondis.\arabic{subsubsection}}

\usepackage{array}
\usepackage{times}
\usepackage{graphicx}
\usepackage{amsmath}
\usepackage{amssymb}
\usepackage{amsfonts}
\usepackage{cite}
\usepackage{setspace}
\usepackage{multido}
\usepackage[justification=Centering]{subfig}
\usepackage{algorithm}
\usepackage{algorithmic}
\usepackage{tikz}
\usepackage{xcolor}
\usepackage{multirow}
\usepackage{paralist}
\usepackage{url}
\usepackage{etoolbox}

\usetikzlibrary{shadows,patterns,shapes,arrows,decorations.pathmorphing,backgrounds,positioning,fit,plotmarks}

\newtheorem{lemma}{Lemma}
\newtheorem{theorem}{Theorem}
\newtheorem{corollary}{Corollary}
\newtheorem{definition}{Definition}

 \def\myendproof{{\ \vbox{\hrule\hbox{%
   \vrule height1.3ex\hskip0.8ex\vrule}\hrule }}\par}
 \newenvironment{proof}{\noindent{\bf Proof. }}{\myendproof}
 \newenvironment{proofAppendix}[1]{\noindent{\bf Proof of #1. }}{\myendproof}

\newcommand{\floor}[1]{\left\lfloor{#1}\right\rfloor}
\newcommand{\ceiling}[1]{\left\lceil{#1}\right\rceil}
\newcommand{\setof}[1]{\left\{{#1}\right\}}

\newcommand{\framework}[1]{$\mathbf{k^2U}$}

\newcommand{\resourceArbM}{3-\frac{1}{M}}

\providebool{techreport}
\setbool{techreport}{true}

\tikzset{
    task/.style={shade, shading=radial, rectangle,minimum height=.1cm,
        inner color=#1!20, outer color=#1!60!gray},
    task1/.style={task=yellow, minimum width=13mm},
    task2/.style={task=orange, minimum width=13mm},
    task3/.style={task=red, minimum width=13mm},
    task4/.style={task=green, minimum width=13mm},
    task5/.style={task=blue, minimum width=13mm},
    task6/.style={task=purple, minimum width=13mm},
    task7/.style={task=cyan, minimum width=13mm},
    task8/.style={task=pink, minimum width=13mm},
}

\floatname{algorithm}{Algorithm}
\newcommand{\algorithmicinput}{\textbf{Input:}}
\newcommand{\INPUT}{\item[\algorithmicinput]}

\ifbool{techreport}{
\pagestyle{plain}
}
{
\pagestyle{empty}
}

\tikzstyle{entity} = [top color=white, bottom color=red!40, 
                            draw=blue!50!black!100, drop shadow]
\tikzstyle{relationship} = [diamond, shape aspect=2, top color=white, bottom color=red!20, 
                                  draw=red!50!black!100, drop shadow]
\tikzstyle{end} = [top color=white, bottom color=blue!30, 
                            draw=blue!50!black!100, drop shadow]

\tikzstyle{max} = [top color=white, bottom color=orange!60, 
                            draw=blue!50!black!100, drop shadow]

\title{Partitioned Multiprocessor Fixed-Priority Scheduling of
  Sporadic Real-Time Tasks}

 \author{Jian-Jia Chen\\
   $^1$Department of Informatics, TU Dortmund University, Germany\\
   E-mail: jian-jia.chen@cs.uni-dortmund.de
 }

\begin{document}

\maketitle

\begin{abstract}
  Partitioned multiprocessor scheduling has been widely accepted in
  academia and industry to statically assign and partition real-time
  tasks onto identical multiprocessor systems. This paper studies
  fixed-priority partitioned multiprocessor scheduling for sporadic
  real-time systems, in which deadline-monotonic scheduling is applied
  on each processor.  Prior to this paper, the best known results are
  by Fisher, Baruah, and Baker with speedup factors $4-\frac{2}{M}$
  and $3-\frac{1}{M}$ for arbitrary-deadline and constrained-deadline
  sporadic real-time task systems, respectively, where $M$ is the
  number of processors.  We show that a
  greedy mapping strategy has a speedup factor $3-\frac{1}{M}$ when
  considering task systems with arbitrary deadlines. Such a factor holds for polynomial-time
  schedulability tests and exponential-time (exact) schedulability
  tests. Moreover, we also improve the speedup factor to $2.84306$
  when considering constrained-deadline task systems. 
We also provide tight examples when the fitting strategy in the
mapping stage is arbitrary and $M$ is sufficiently large.
 For both constrained- and
  arbitrary-deadline task systems,  the analytical result surprisingly shows that using exact tests does not gain
  theoretical benefits (with respect to speedup factors) if the speedup factor analysis
 is oblivious of the particular fitting strategy used.
\end{abstract}


\section{Introduction}
\label{sec:intro}

The sporadic task model has been widely adopted as the basic model for
real-time systems with recurring executions
\cite{Mok:1983:FDP:888951}.  A sporadic real-time task $\tau_i$ is
characterized by its \emph{minimum inter-arrival time} $T_i$, its
timing constraint or \emph{relative deadline} $D_i$, and its
(worst-case) \emph{execution time} $C_i$. A sporadic task defines an
infinite sequence of task instances, also called \emph{jobs}, that
arrive with the minimum inter-arrival time constraint. Under the
minimum inter-arrival time constraint, any two consecutive jobs of
task $\tau_i$ should be temporally separated by at least $T_i$.  When
a job of task $\tau_i$ arrives at time $t$, the job should finish no
later than its \emph{absolute deadline} $t+D_i$. If we consider a task
releases its jobs periodically, such a task model is the well-known
Liu and Layland task model \cite{liu73scheduling}, where $T_i$ is the
\emph{period} of the task.  An input task set is said to have
1)~\emph{implicit deadlines} if the relative deadlines of sporadic
tasks are equal to their minimum inter-arrival times,
2)~\emph{constrained deadlines} if the minimum inter-arrival times
are no less than their relative deadlines, and (3)~\emph{arbitrary
  deadlines}, otherwise.

Through this paper, we only consider implicitly preemptive scheduling. That is, a job
may be preempted by another job on a processor.
For scheduling sporadic tasks on a processor, the preemptive
earliest-deadline-first (EDF) policy is optimal \cite{liu73scheduling}
to meet the timing constraints. However, EDF requires to prioritize
the jobs in the ready queue by using their absolute deadlines, and the
overhead is in general not negligible. The industrial practice is to
use fixed-priority scheduling, also supported in most real-time
operating systems, in which a task is assigned with a fixed priority level.
The seminal work by Liu and Layland \cite{liu73scheduling} shows that
rate monotonic (RM) scheduling is optimal for uniprocessor fixed-priority
scheduling when considering implicit-deadline task systems. Moreover,
deadline monotonic (DM) scheduling \cite{journals/pe/LeungW82} is
optimal for uniprocessor fixed-priority scheduling for constrained-deadline task
systems. For arbitrary-deadline task systems, Audsley et
al. \cite{audsley1993applying} provide an optimal priority assignment
algorithm to define the priority levels of the sporadic tasks for
uniprocessor fixed-priority scheduling.

Testing whether a task set can be feasibly scheduled by a scheduling
algorithm is called a \emph{schedulability test}. Even though RM and
DM are known to be optimal for uniprocessor fixed-priority scheduling, 
the exact schedulability tests for uniprocessor fixed-priority
scheduling requires pseudo-polynomial time by using
the exact tests by Lehoczky, Sha, and Ding
\cite{DBLP:conf/rtss/LehoczkySD89} for constrained-deadline systems
and in exponential time by Lehoczky \cite{DBLP:conf/rtss/Lehoczky90} for
arbitrary-deadline systems. Specifically, computing the
worst-case response time of one (lowest-priority) task is shown
 ${\cal NP}$-hard by Eisenbrand
and Rothvo{\ss} \cite{DBLP:conf/rtss/EisenbrandR08}. There have been
extensive results about testing the schedulability of uniprocessor fixed-priority
scheduling. 
The more efficient strategy is to provide
only sufficient conditions that can be verified in polynomial-time,
like the utilization bound
\cite{liu73scheduling,journals/tc/LeeSP04,DBLP:dblp_journals/tc/BurchardLOS95},
the quadratic utilization bound \cite{DBLP:journals/tc/Bini15}, the
hyperbolic utilization bound
\cite{bini2003rate,ChenHLRTSS2015}, the
approximated request bound functions
\cite{DBLP:conf/ecrts/AlbersS04,DBLP:conf/rtss/ChakrabortyKT02,DBLP:conf/ecrts/FisherB05}.

\begin{table*}[t]
  \centering
\renewcommand{\arraystretch}{1.3}
  \begin{tabular}{|l|p{2.9cm}|p{3.5cm}|l|p{3.5cm}|}
    \hline
    & implicit deadlines & constrained deadlines& arbitrary deadlines \\
    \hline
    \multirow{2}{*}{partitioned EDF} &
    $\frac{4}{3}-\frac{1}{3M}$ \cite{DBLP:journals/siamam/Graham69} &
    $3-\frac{1}{M}$ \cite{DBLP:journals/tc/BaruahF06} & $4-\frac{2}{M}$ \cite{DBLP:conf/rtss/BaruahF05}\\
     & PTAS \cite{DBLP:journals/jacm/HochbaumS87}& $2.6322-\frac{1}{M}$ \cite{Chakraborty2011a}& $3-\frac{1}{M}$ \cite{Chakraborty2011a}\\
    \hline 
    \multirow{2}{*}{partitioned DM} & $\frac{7}{4}$
    \cite{DBLP:dblp_journals/tc/BurchardLOS95} & $3-\frac{1}{M}$
    \cite{DBLP:conf/ecrts/FisherBB06} & $4-\frac{2}{M}$ \cite{DBLP:conf/ecrts/FisherBB06}\\
    & 1.5 \cite{DBLP:conf/waoa/KarrenbauerR10}& $2.84306$ (\emph{this
     paper}) & $3-\frac{1}{M}$ (\emph{this paper}) \\
    \hline 
      \hline
  \end{tabular}
  \caption{Summary of the speedup factors in the multiprocessor
    partitioned scheduling problem for sporadic task systems. For more details of
    implicit-deadline cases, please refer to Table III in the survey \cite{DBLP:journals/csur/DavisB11}.}
  \label{tab:summary}
\end{table*}

To quantify the performance loss due to efficient schedulability tests
and assigning tasks with fixed priority levels, we will adopt the
notion of speedup factors, (also known as resource augmentation
factors).  A fixed-priority scheduling algorithm with a speedup
factor $\rho$ guarantees to produce feasible schedules by running (each processor)
$\rho$ times as fast as in the original platform (speed), if there
exists a feasible schedule (under arbitrary policies) for the task
system. 
The speedup factors of DM
scheduling, with respect to the optimal uniprocessor EDF scheduling,
are $\frac{1}{\ln 2}$, $1.76322$, and $2$ for implicit-deadline,
constrained-deadline, and arbitrary-deadline task sets
\cite{DBLP:journals/rts/DavisRBB09,conf:/rtns09/Davis}, respectively.

To schedule real-time tasks on multiprocessor platforms, there have
been three widely adopted paradigms: partitioned, global, and
semi-partitioned scheduling. A comprehensive survey of multiprocessor
scheduling in real-time systems can be found in
\cite{DBLP:journals/csur/DavisB11}. In this paper, we consider
\emph{partitioned scheduling}, in which the tasks are statically
partitioned onto processors and all the processors are identical. That
is, all the jobs of a task are executed on a specific processor with
fixed-priority scheduling.

However, problems on multiprocessors become ${\cal NP}$-complete (or
worse) in the strong sense even in the simplest possible cases.  For
example, deciding if an implicit-deadline task set with the same
period is schedulable on multiple processors is already ${\cal
  NP}$-complete in the strong sense \cite{Mok:1983:FDP:888951}.  To
cope with these ${\cal NP}$-hardness issues, one natural approach is
to focus on approximation algorithms, i.e., polynomial time algorithms
that produce an approximate solution instead of an exact one. There
have been many results for implicit-deadline task systems, as
summarized in the survey paper
\cite{DBLP:journals/csur/DavisB11}. But, only a few results are known
for constrained-deadline and arbitrary-deadline task systems.

When considering sporadic task sets with constrained or arbitrary
deadlines, the problem becomes more complicated, when EDF or
fixed-priority scheduling is adopted on a processor. The recent
studies in \cite{BaruahRTSS2011,ChenECRTS12} provide polynomial-time
approximation schemes for some special cases when speeding-up is
adopted for EDF scheduling. For general cases, Baruah and Fisher
\cite{DBLP:conf/rtss/BaruahF05,DBLP:journals/tc/BaruahF06} propose a
simple method, denoted as \emph{deadline-monotonic partitioning} in
this paper, which 1) considers the tasks in a non-decreasing order of
their relative deadlines, and 2) assigns a task (in the above order)
to a processor if it can pass the schedulability condition. If there
are multiple processors that are feasible for assigning a task, the
deadline-monotonic partitioning algorithm by Baruah and Fisher
\cite{DBLP:conf/rtss/BaruahF05,DBLP:journals/tc/BaruahF06} uses the
first-fit strategy, but the analysis works for any arbitrary fitting
strategy. The (theoretical) advantage of the first-fit strategy was not shown in
the literature when we consider constrained- or arbitrary-deadline
task systems.

The deadline-monotonic partitioning strategy is simple, but has been
shown effective in the literature
\cite{DBLP:conf/rtss/BaruahF05,DBLP:journals/tc/BaruahF06,Chakraborty2011a,DBLP:conf/ecrts/FisherBB06}.
When adopting speeding-up for resource augmentation, by using EDF on a
processor, the deadline-monotonic partitioning proposed by Baruah and
Fisher \cite{DBLP:conf/rtss/BaruahF05,DBLP:journals/tc/BaruahF06} has
been shown with a $\resourceArbM$ speedup factor by Chen and
Chakraborty \cite{Chakraborty2011a}, where $M$ is the given number of
identical processors.  Prior to this paper, for fixed-priority
multiprocessor partitioned scheduling for constrained- and
arbitrary-deadline task systems, the best known results are by Fisher,
Baruah, and Baker \cite{DBLP:conf/ecrts/FisherBB06} with speedup
factors $4-\frac{2}{M}$ and $3-\frac{1}{M}$ for arbitrary-deadline and
constrained-deadline sporadic real-time task systems,
respectively. All the above results are based on a
linear-approximation to efficiently and safely test the schedulability
under EDF or DM scheduling to decide whether a task can be assigned on
a processor.

{\bf Our Contributions:} Table~\ref{tab:summary} summarizes the
related results and the contribution of this paper for multiprocessor
partitioned scheduling. We focus on fixed-priority multiprocessor
partitioned scheduling, and improve the best known results by Fisher,
Baruah, and Baker \cite{DBLP:conf/ecrts/FisherBB06}. The
deadline-monotonic partitioning algorithm is explored in a great
detail in this paper. Our contributions are:
\begin{itemize}
\item We show that the deadline-monotonic partitioning algorithm has a
  speedup factor $3-\frac{1}{M}$ when considering task systems with
  arbitrary deadlines, where $M$ is the number of processors. Such a
  factor holds for polynomial-time schedulability tests and
  exponential-time (exact) schedulability tests. Moreover, we
  also improve the speedup factor to $2.84306$ when considering
  constrained-deadline task systems by using polynomial-time
  and pseudo-polynomial-time (exact)
  schedulability tests.
\item The existing results by adopting the deadline-monotonic
  partitioning algorithm were analyzed based on approximated
  schedulability tests. One of our key contributions is to answer the
  question: \emph{Will it be possible to further reduce the speedup factors by
    using exact tests in the deadline-monotonic partitioning
    algorithm?} Our answer to this question is \emph{\bf NO!!} Using
  exact tests in the above algorithm does not have any chance to
  reduce the speedup factors  if the speedup factor analysis
 is oblivious of the particular fitting strategy used.
 We show that all the speedup factor
  analyses in this paper are asymptotically tight with polynomial-time
  schedulability tests and exponential-time (or pseudo polynomial-time) schedulability tests
  if the speedup factor analysis is oblivious of the particular fitting strategy used. As a result, to improve the
  speedup factor, better fixed-priority scheduling strategies or more
  precise analyses for concrete fitting strategies are needed.
\end{itemize}

\section{System Models and Preliminary Results}
\label{sec:model}

\subsection{Task and Platform Model}

We consider a set ${\bf T}=\setof{\tau_1, \tau_2, \ldots, \tau_N}$ of
$N$ independent sporadic real-time tasks.  A task $\tau_i$ is defined
by $(C_i, T_i, D_i)$. That is, for task $\tau_i$, $D_i$ is its
relative deadline, $T_i$ is its minimum inter-arrival time (period),
and $C_i$ is its (worst-case) execution time.  We consider identical
processors in the platform. Therefore, no matter which processor a
task is assigned to, the execution and timing property
remains. According to the relations of the relative deadlines and the
minimum inter-arrival times of the tasks in ${\bf T}$, the task set
can be identified to be with 1) implicit deadlines, i.e., $D_i=T_i,
\forall \tau_i \in {\bf T}$, 2) constrained deadlines, i.e., $D_i
\leq T_i, \forall \tau_i \in {\bf T}$, or 3) arbitrary deadlines.

For brevity, the \emph{utilization} of task $\tau_i$ is denoted by
$U_i=\frac{C_i}{T_i}$. Moreover, let $\Delta_i$ be $\max\{U_i,
\frac{C_i}{D_i}\}$. For a set ${\bf X}$, its cardinality is denoted by
$|{\bf X}|$.

We will consider preemptive fixed-priority
scheduling on each processor. Specifically, we will only use
deadline-monotonic (DM) scheduling on each processor to
assign the priority levels of the tasks. That is, task $\tau_i$ is
with higher priority than $\tau_j$ if $D_i < D_j$, in which the ties
are broken arbitrarily.  
Therefore, for the rest of this paper, we index the tasks from the 
shortest relative deadline to the longest, i.e., $D_i \leq D_j$ if $i 
< j$. 
Note that DM priority assignment is an optimal
fixed-priority scheduling for implicit-deadline and constrained-deadline
task sets \cite{journals/pe/LeungW82}.  
Although DM priority assignment is not an optimal fixed-priority
assignment when we consider arbitrary-deadline task systems, it has
a constant speedup factor, to be discussed in Section~\ref{sec:def-speedup}.

\subsection{Problem Definition}

Given task set ${\bf T}$, a \emph{feasible task partition} on $M$
identical processors is a collection of $M$ subsets, says, ${\bf T}_1,
{\bf T}_2, \ldots, {\bf T}_M$, of ${\bf T}$ such that
\begin{compactitem}
\item ${\bf T}_m \cap {\bf T}_{m'}=\emptyset$ for all $m\neq m'$,
\item $\cup_{m=1}^{M} {\bf T}_m$ is equal to the input task set ${\bf
    T}$, and
\item set ${\bf T}_m$ can meet the timing constraints by DM
  scheduling on a processor $m$.
\end{compactitem}

Without loss of generality, we can assume that  $U_i \leq 100\%$ and
$\frac{C_i}{D_i} \leq 100\%$, i.e., $\Delta_i \leq 100\%$, for any task $\tau_i$;
otherwise, there is clearly no feasible task partition.

\subsection{Speedup Factors/Bounds}
\label{sec:def-speedup}
This paper focuses on the case where the arrival times of the sporadic
tasks are not specified. Therefore, the approximation is for the worst
cases by considering the worst-case behaviour to be feasibly scheduled
by DM.  If an algorithm ${\cal A}$ for the studied problem has a
\emph{speedup factor} $\rho$, it guarantees to \emph{always produce a feasible solution by speeding each
  processor up to $\rho$ times of the original
  speed in the platform, if task set ${\bf T}$  can be
  feasibly scheduled (not restricted to DM) on
  the original $M$ identical processors}. In other words, by taking
the negation of the above statement, we know that \emph{if the
  algorithm ${\cal A}$ fails to feasibly partition the task set ${\bf
    T}$ on $M$ identical processors, there is no feasible task
  partition when each processor runs $\frac{1}{\rho}$ times slower
  than the original platform speed.} For the rest of this paper, we
use $1$ to denote the original platform speed. Therefore, running the
platform at speed $s$ implies that the execution time of task $\tau_i$
becomes $\frac{C_i}{s}$. Note that speedup factors are used for
quantifying the behaviour of the designed algorithm. This is useful,
especially for the negation part to quantify the error the algorithm
makes when it does not provide a feasible solution.

For fixed-priority scheduling, the speedup factors of DM scheduling,
with respect to the optimal uniprocessor EDF scheduling, are
$\frac{1}{\ln 2}$, $1.76322$, and $2$ for implicit-deadline,
constrained-deadline, and arbitrary-deadline task sets
\cite{conf:/rtns09/Davis}, respectively.  To quantify the
schedulability of the input task set, we would need to know the
necessary condition for being schedulable at speed $s$ on the $M$
processors. The necessary conditions $\max_{\tau_i \in {\bf T}}
\Delta_i \leq s$ and $\sum_{\tau_i \in {\bf T}} \frac{U_i}{M} \leq s$
are pretty straightforward. As we focus on arbitrary-deadline and
constrained-deadline sporadic task systems, we can also quantify the
necessary condition defined by the demand. Here, we can release the
first job of tasks synchronously (say, at time $0$), and the
subsequent job arrivals should be as rapidly as legally possible. A
necessary condition to be schedulable is to ensure that the total
execution time of the jobs arriving no earlier than $a$ and with
relative deadlines no later than $d$ is no more than $M\cdot (d-a)$ for any
$a < d$.  This is identical to the well-known \emph{demand bound
  function} ${\sc dbf}(\tau_i, t)$, as in
\cite{DBLP:conf/rtss/BaruahMR90}, of a task $\tau_i$ within any time
interval with length equal to $t$, defined as
\begin{equation}
  \label{eq:dbf}
  {\sc dbf}(\tau_i, t) = \max\left\{0, \floor{\frac{t-D_i}{T_i}}+1\right\}\times C_i.
\end{equation}
Therefore, as a necessary condition, to
ensure the schedulability on $M$ processors, if a task set is
schedulable for an algorithm on $M$ processors, then
\begin{equation}
  \label{eq:dbf-necessary}
  \forall t > 0, \qquad\sum_{\tau_i \in {\bf T}} {\sc dbf}(\tau_i, t) \leq M t.
\end{equation}

With the above discussions, we can conclude the following lemma for
the necessary condition to be schedulable by any algorithm, which has
also been utilized by Chen and Chakraborty \cite{Chakraborty2011a}.
\begin{lemma}
  \label{lemma:lower-speed-bound}
  A task set is not schedulable by any multiprocessor scheduling
  algorithm by running the $M$ processors at
  any speed $s$ if
  \begin{equation}
    \label{eq:def-s-overall}
    \max\left\{\max_{t> 0}\frac{\sum_{\tau_i \in {\bf T}} dbf(\tau_i, t)}{Mt}, \frac{\sum_{\tau_i \in {\bf T}} U_i}{M},\max_{\tau_i \in {\bf T}} \Delta_i \right\} > s.    
  \end{equation}
\end{lemma}

For the rest of the paper, we will focus ourselves on the negation
part of the speedup factor analysis. That is, we are only interested
in the failure cases of the partitioning algorithm and use
Lemma~\ref{lemma:lower-speed-bound} to quantify $s$ for showing the
speedup factors.
Note that Lemma~\ref{lemma:lower-speed-bound} is also the necessary
condition for global multiprocessor scheduling. It may seem that we
are more pessimistic by comparing to the necessary condition of global
multiprocessor scheduling instead of that of partitioned
multiprocessor scheduling. However, in our tightness analysis,
comparing to partitioned scheduling and global scheduling does not
differ very much.

\section{Deadline-Monotonic Partitioning}
\label{sec:deadline-monotonic}

This section presents the deadline-monotonic partitioning strategy,
proposed by Baruah and Fisher
\cite{DBLP:journals/tc/BaruahF06,DBLP:conf/rtss/BaruahF05,DBLP:conf/ecrts/FisherBB06},
for the multiprocessor partitioned scheduling problem. Note that such
a strategy works in general for fixed-priority scheduling (RM, DM) and
dynamic-priority scheduling (EDF), by adopting proper schedulability
tests. The speedup factor for EDF/DM was shown to be $3-\frac{1}{M}$
and $4-\frac{2}{M}$
\cite{DBLP:journals/tc/BaruahF06,DBLP:conf/rtss/BaruahF05,DBLP:conf/ecrts/FisherBB06}
for constrained-deadline systems and arbitrary-deadline systems,
respectively. When considering EDF scheduling, Chen and Chakraborty
\cite{DBLP:journals/rts/ChenC13} improved the speedup factor to
$2.6322-\frac{1}{M}$ and $3-\frac{1}{M}$ for constrained-deadline
systems and arbitrary-deadline systems, respectively.

\begin{algorithm}[t]
  \caption{Deadline-Monotonic Partitioning}
  \label{alg:dmp}
  \begin{algorithmic}[1]
\footnotesize

    \INPUT set ${\bf T}$ of $N$ tasks, $M$ processors;

    \STATE re-index (sort) tasks such that $D_i \leq D_j$ for $i < j$;

    \STATE  ${\bf T}_1 \leftarrow \setof{\tau_1}$;${\bf T}_m
    \leftarrow \emptyset, \forall m=2,3,\ldots,M$;

    \FOR {$k=2$ to $N$}

    \IF {$\exists m \in \setof{1, 2, \ldots, M}$ such that ${\bf T}_m\cup \setof{\tau_k}$ is
      schedulable by DM fixed-priority scheduling} \label{algo:fitting-test}

    \STATE choose $m \in \setof{1, 2, \ldots, M}$ \emph{\bf by
      preference} such that ${\bf T}_m\cup \setof{\tau_k}$ is
      schedulable by DM fixed-priority scheduling; \label{algo:fitting}

    \STATE assign $\tau_k$ to processor $m$ with ${\bf T}_m \leftarrow {\bf T}_m
    \cup \setof{\tau_k}$;

    \ELSE
    \STATE return ``no feasible schedule is found'';
    \ENDIF

    \ENDFOR

    \STATE return feasible task partition ${\bf T}_1, {\bf T}_2,
    \ldots, {\bf T}_M$;
  \end{algorithmic}
\end{algorithm}

For completeness, we repeat the algorithm in
\cite{DBLP:journals/tc/BaruahF06,DBLP:conf/rtss/BaruahF05,DBLP:conf/ecrts/FisherBB06},
in which the pseudo-code is presented in Algorithm~\ref{alg:dmp}.
Deadline-monotonic partitioning considers the given tasks from the shortest
relative deadline to the longest relative deadline for
assignment. When a task $\tau_k$ is considered, a processor $m$ with
$m \in \setof{1,2,\ldots,M}$ is selected to assign task $\tau_k$, where
${\bf T}_m$ is the set of the tasks (as a subset of $\setof{\tau_1,
  \tau_2, \ldots, \tau_{k-1}}$), which have been assigned to processor
$m$ before considering $\tau_k$. If there is no feasible $m$ that can
feasibly schedule $\tau_k$ and ${\bf T}_m$ on the processor, we return
that no feasible solution is found by this algorithm.

\subsection{Fitting Strategy} 
\label{sec:fitting}

The fitting strategy when we consider to assign task $\tau_k$ on a processor $m$ can be
\begin{compactitem}
\item the \emph{first-fit} strategy: by choosing the minimum $m$ that is feasible;
\item the \emph{arbitrary-fit} strategy: by choosing any $m$ that is feasible (this is also the case if the speedup factor analysis
is oblivious of the particular fitting strategy used); 
\item the \emph{best-fit} strategy: by choosing the index $m$ that has the maximum
  \emph{workload-index};
\item the \emph{worst-fit} strategy: by choosing the index $m$ that has the minimum
  \emph{workload-index}.  
\end{compactitem}
The workload-index can be defined as the total utilization or other
means. The analysis in the literature
\cite{DBLP:journals/tc/BaruahF06,DBLP:conf/rtss/BaruahF05,DBLP:conf/ecrts/FisherBB06,DBLP:journals/rts/ChenC13}
works in general by using any fitting strategy listed above, even
though in several cases only the first-fit strategy was mentioned in
the descriptions
\cite{DBLP:journals/tc/BaruahF06,DBLP:conf/rtss/BaruahF05,DBLP:conf/ecrts/FisherBB06}.

\subsection{Schedulability Tests for DM}
\label{sec:test-for-dm}

Therefore, the remaining building block is to test whether task
$\tau_k$ can be feasibly scheduled on a processor $m$ under DM
scheduling. This has been widely studied in the literature. We will
review some of these methods and explain their corresponding speedup factors
when they are adopted in Step \ref{algo:fitting-test} in
Algorithm~\ref{alg:dmp}.

\noindent {\bf Constrained Deadline:}
To verify the schedulability of constrained-deadline task $\tau_k$
under fixed-priority scheduling in uniprocessor systems, the
time-demand analysis (TDA) \cite{DBLP:conf/rtss/LehoczkySD89} can be
adopted. That is, if and only if
\begin{equation}
  \label{eq:exact-test-constrained-deadline}
\exists t \mbox{ with } 0 < t \leq D_k {\;\; and \;\;} C_k +
\sum_{\tau_i \in {\bf T}_m} \ceiling{\frac{t}{T_i}}C_i \leq t,
\end{equation}
then task $\tau_k$ is schedulable under DM scheduling, where ${\bf
  T}_m$ is the set of tasks with higher priority than task $\tau_k$
since we sort the tasks according to their relative deadlines. TDA
requires pseudo-polynomial-time complexity, as all the points that lie
in $(0, D_k]$ need to be checked for
Eq.~\eqref{eq:exact-test-constrained-deadline}.

Fisher, Baruah, and Baker \cite{DBLP:conf/ecrts/FisherBB06}
approximate the test in Eq.~\eqref{eq:exact-test-constrained-deadline}
by testing only
\begin{equation}
  \label{eq:approxi-test-constrained-deadline}
  \exists t \mbox{ with } 0 < t \leq D_k {\;\; and \;\;} C_k +
  \sum_{\tau_i \in {\bf T}_m} \left(1+\frac{t}{T_i}\right)C_i \leq t.
\end{equation}
Due to the linearity of the condition in
Eq.~\eqref{eq:approxi-test-constrained-deadline}, the test is
equivalent to the verification of whether 
\begin{align}
  C_k + \sum_{\tau_i \in {\bf T}_m} \left(1+\frac{D_k}{T_i}\right)C_i
  \leq D_k \label{eq:fbb-constrained-1}
\end{align}  
for constrained-deadline systems.

We can also approximate the schedulability test by using
utilization-based analysis as follows: We classify the task
set ${\bf T}_m$ into two subsets:
\begin{itemize}
\item ${\bf T}_m^1$ consists of the higher-priority tasks with periods
  smaller than $D_k$.
\item ${\bf T}_m^2$ consists of the higher-priority tasks with periods
  larger than or equal to $D_k$.
\end{itemize}
The following theorem has been concluded recently by using a
utilization-based schedulability-test framework \cite{ChenHLRTSS2015}.

\begin{theorem}[Chen, Huang, Liu \cite{ChenHLRTSS2015}]
\label{theorem:sporadic-general}
Task $\tau_k$ in a sporadic task system with constrained deadlines is
schedulable by DM scheduling algorithm on processor $m$ if
\begin{equation}
\label{eq:schedulability-sporadic-any-a}
(\frac{C_{k,m}'}{D_k}+1) \prod_{\tau_j \in {\bf T}_m^1} (U_j + 1)\leq 2.
\end{equation}
where 
$C_{k,m}'$ is $C_k + \sum_{\tau_i \in {\bf T}_m^2} C_i$.
\end{theorem}

\noindent{\bf Arbitrary Deadline:}
For arbitrary-deadline systems, the exact schedulability test is to
use a \emph{busy-window} concept to evaluate the worst-case response
time \cite{DBLP:conf/rtss/Lehoczky90} by using TDA.  The
finishing time $R_{k,h}$ of the $h$-th job of task $\tau_k$ in the
busy window is the minimum $t$ such that
\[ 
h C_k + \sum_{\tau_i \in {\bf T}_m} \ceiling{\frac{t}{T_i}}C_i \leq t.
\] 
Therefore, its response time is $R_{k,h}-(h-1)T_k$. To test the busy
window length of task $\tau_k$, the busy window of task $\tau_k$
finishes on the $h$-th job if $R_{k,h} \leq h T_k$. The maximum
response time among the jobs in the busy window is the worst-case
response time \cite{DBLP:conf/rtss/Lehoczky90}. The schedulability test is
exact, but this takes exponential time complexity, since the busy window length can be up to the hyper-period, which is exponential of the input size.

The approximation by Fisher, Baruah, and Baker
\cite{DBLP:conf/ecrts/FisherBB06} in
Eq.~\eqref{eq:approxi-test-constrained-deadline} remains feasible for
analyzing the arbitrary-deadline systems. This leads to test whether 
\begin{subequations}
\begin{align}
   C_k + \sum_{\tau_i \in {\bf
    T}_m} \left(1+\frac{D_k}{T_i}\right)C_i \leq D_k & \;\;\;\mbox{ and }   \label{eq:fbb-arbitrary-1}\\
  U_k + \sum_{\tau_i \in {\bf
    T}_m} U_i \leq 1.   \label{eq:fbb-arbitrary-2}
\end{align}  
\end{subequations}
Eq.~\eqref{eq:fbb-arbitrary-2} is important in arbitrary-deadline
systems to ensure that the approximation in
Eq.~\eqref{eq:fbb-arbitrary-1} does not underestimate the workload
after $D_k$. 

Moreover, Bini et al. \cite{DBLP:journals/tc/BiniNRB09} provide a
tighter analysis than Eq.~\eqref{eq:fbb-arbitrary-1}. They show that
the worst-case response time of task $\tau_k$ is at most 
\[
\frac{C_k+ \sum_{\tau_i \in {\bf T}_m} C_i - \sum_{\tau_i \in {\bf
      T}_m } U_i C_i}{1-\sum_{\tau_i \in {\bf T}_m} U_i}.
\]
Therefore, the schedulability condition in
Eqs. \eqref{eq:fbb-arbitrary-1} and \eqref{eq:fbb-arbitrary-2}
can be rewritten as
\begin{subequations}
\begin{align}
  C_k + D_k(\sum_{\tau_i \in {\bf T}_m} U_i)+ \sum_{\tau_i \in {\bf 
      T}_m} C_i - \sum_{\tau_i \in {\bf T}_m } U_i C_i\leq 
  D_k,  \label{eq:fbb-arbitrary-3}\\
 U_k + \sum_{\tau_i \in {\bf 
    T}_m} U_i \leq 1.   \label{eq:fbb-arbitrary-4}
\end{align}    
\end{subequations}

\subsection{Time Complexity and Correctness}

Since we use partitioned scheduling, as long as
Algorithm~\ref{alg:dmp} returns a task partition, it is guaranteed to
be feasible if the test in Step \ref{algo:fitting-test}
is a sufficient schedulability test for task $\tau_k$ by using DM scheduling. The time complexity of the algorithm
depends upon the time complexity of the schedulability test and the
fitting strategy. Suppose that the fitting strategy requires time
complexity $O(F)$ for one task and the time complexity to test whether
task $\tau_k$ is schedulable on processor $m$ is $O(H)$.  Then, the
overall time complexity is $O(NMH + NF)$. All the fitting strategies
mentioned in Section \ref{sec:fitting} are in polynomial time.  Since
$H$ can be polynomial, pseudo polynomial, or exponential of the input size, the time
complexity is dominated by the adopted schedulability test algorithm.

The main issue here is to answer what can be guaranteed when
Algorithm~\ref{alg:dmp} returns failure in task partitioning. We will
quantify such failures by showing that $s$ in
Eq.~\eqref{eq:def-s-overall} is also sufficiently large to provide the
speedup factor guarantee (by using the negation arguments).

\section{Analysis for Arbitrary-Deadline Systems}
\label{sec:arbitrary}

This section presents our analysis for arbitrary-deadline systems.
Our analysis is similar to the analysis by Fisher, Baruah, and Baker
\cite{DBLP:conf/ecrts/FisherBB06}, but is tighter. Here, we will
mainly analyze the property by using the schedulability condition in
Eqs.~\eqref{eq:fbb-arbitrary-1} and~\eqref{eq:fbb-arbitrary-2}. At the
end of this section, we will explain why the analysis also works for
arbitrary-deadline TDA analysis by Lehoczky
\cite{DBLP:conf/rtss/Lehoczky90} and response time analysis by Bini et
al. \cite{DBLP:journals/tc/BiniNRB09}.

\begin{theorem}
\label{theorem:sporadic-arbitrary}
The speedup factor of Algorithm~\ref{alg:dmp} is $3-\frac{1}{M}$ when
adopting
Eqs.~\eqref{eq:fbb-arbitrary-1}~and~\eqref{eq:fbb-arbitrary-2} for DM
schedulability test under any fitting strategy.
\end{theorem}
\begin{proof}
Suppose that Algorithm~\ref{alg:dmp} fails to find a feasible
assignment for task $\tau_k$ due to the failure when testing
Eq.~\eqref{eq:fbb-arbitrary-1} or Eq.~\eqref{eq:fbb-arbitrary-2}. Let
${\bf M}_1$ be the set of processors in which
Eq.~\eqref{eq:fbb-arbitrary-1} fails. Let ${\bf M}_2$ be the set of
processors in which Eq.~\eqref{eq:fbb-arbitrary-1} succeeds but
Eq.~\eqref{eq:fbb-arbitrary-2} fails. Since task $\tau_k$ cannot be
assigned on any of the $M$ processors, we have $|{\bf M}_1|+|{\bf
  M}_2|=M$. By the violation of Eq.~\eqref{eq:fbb-arbitrary-1}, we
know that
\begin{align}
&  |{\bf M}_1|C_k + \sum_{m \in {\bf M}_1}\sum_{\tau_i \in {\bf
      T}_m}\left(1+\frac{D_k}{T_i}\right)C_i > |{\bf
    M}_1|D_k    \nonumber\\
\Rightarrow & |{\bf M}_1|\frac{C_k}{D_k} + \sum_{m \in {\bf M}_1}\sum_{\tau_i \in {\bf
      T}_m}\left(\frac{C_i}{D_k}+U_i\right) > |{\bf
    M}_1|. \label{eq:fbb-arbitrary-violate-1}
\end{align}
By the violation of Eq.~\eqref{eq:fbb-arbitrary-2}, we
know that
\begin{equation}
  \label{eq:fbb-arbitrary-violate-2}
  |{\bf M}_2|U_k + \sum_{m \in {\bf M}_2}\sum_{\tau_i \in {\bf
      T}_m} U_i > |{\bf M}_2|.
\end{equation}
Recall that $\Delta_k$ is defined as $\max\{U_k,
\frac{C_k}{D_k}\}$. By Eqs.~\eqref{eq:fbb-arbitrary-violate-1}
and~\eqref{eq:fbb-arbitrary-violate-2}, we know that
{\small \begin{align}
&  |{\bf M}_1|\frac{C_k}{D_k} +  |{\bf M}_2|U_k + \sum_{i=1}^{k-1}U_i + \sum_{m \in {\bf M}_1}\sum_{\tau_i \in {\bf
      T}_m}\frac{C_i}{D_k} > M\nonumber\\
\Rightarrow \;\;\;& M \Delta_k  + \sum_{i=1}^{k-1}U_i +
\sum_{i=1}^{k-1}\frac{C_i}{D_k} > M. \nonumber\\
\Rightarrow \;\;\;& (M-1)\Delta_k  + \sum_{i=1}^{k}U_i +
\sum_{i=1}^{k}\frac{C_i}{D_k} > M. \nonumber\\
\Rightarrow \;\;\;& (1-\frac{1}{M})\Delta_k  + \frac{\sum_{i=1}^{k}U_i}{M} +
\sum_{i=1}^{k}\frac{C_i}{M D_k} > 1. \nonumber\\
\Rightarrow \;\;\;& (3-\frac{1}{M})\max\left\{\Delta_k, \frac{\sum_{i=1}^{k}U_i}{M},
\sum_{i=1}^{k}\frac{C_i}{M D_k}\right\} > 1. \label{eq:arbitrary-final}
\end{align}
}
Therefore, we know that
 either $\Delta_k > \frac{1}{3-\frac{1}{M}}$,
or $\frac{\sum_{i=1}^{k}U_i}{M} > \frac{1}{3-\frac{1}{M}}$, or
$\sum_{i=1}^{k}\frac{C_i}{M D_k} > \frac{1}{3-\frac{1}{M}}$. Either of
the former two cases implies the unschedulability of any scheduling
algorithm with speed $\frac{1}{3-\frac{1}{M}}$.
The demand bound function at time $D_k$ is $\sum_{i=1}^{N} dbf(\tau_i,
D_k) \geq \sum_{i=1}^{k} C_i$. Therefore, by
Eq.~\eqref{eq:dbf-necessary}, we know that the condition
$\sum_{i=1}^{k}\frac{C_i}{M D_k} > \frac{1}{3-\frac{1}{M}}$ implies
the unschedulability of any scheduling algorithm with speed
$\frac{1}{3-\frac{1}{M}}$. 

Therefore, by using Lemma~\ref{lemma:lower-speed-bound}, we reach the conclusion.
\end{proof}

The following corollaries show that the speedup factor $3-\frac{1}{M}$
holds for any schedulability tests discussed in
Section~\ref{sec:deadline-monotonic} for arbitrary-deadline sporadic
task systems under DM scheduling.

\begin{corollary}
\label{cor:sporadic-arbitrary-poly}
The speedup factor of Algorithm~\ref{alg:dmp} is $3-\frac{1}{M}$ when
adopting
Eqs.~\eqref{eq:fbb-arbitrary-3}~and~\eqref{eq:fbb-arbitrary-4} for
DM schedulability test under any fitting strategy.
\end{corollary}
\begin{proof}
  By not considering the term $- \sum_{\tau_i \in {\bf T}_m } U_i
  C_i$ in Eq.~\eqref{eq:fbb-arbitrary-3}, the violation of Eq.~\eqref{eq:fbb-arbitrary-3} leads to the
  same conclusion in
  Eq.~\eqref{eq:fbb-arbitrary-violate-1}. Therefore, the speedup
  factor remains $3-\frac{1}{M}$.
\end{proof}

\begin{corollary}
\label{cor:sporadic-arbitrary-arbitrary}
The speedup factor of Algorithm~\ref{alg:dmp} is $3-\frac{1}{M}$ when
adopting the exact schedulability test for DM scheduling under any
fitting strategy.
\end{corollary}
\begin{proof}
  If task $\tau_k$ cannot pass the exact schedulability test, it also
  does not pass the sufficient test by using
  Eqs.~\eqref{eq:fbb-arbitrary-1} and
  \eqref{eq:fbb-arbitrary-2}. Therefore, we reach the same conclusion. 
\end{proof}

\subsection{Tightness Analysis} 
The following theorem shows that the analysis in
Theorem~\ref{theorem:sporadic-arbitrary} is asymptotically tight even
for implicit-deadline systems with first-fit strategy.

\begin{theorem}
\label{theorem:sporadic-arbitrary-tight1}
The speedup factor of Algorithm~\ref{alg:dmp} is at least
$3-\frac{3}{M+1}-\gamma$ when adopting
Eqs.~\eqref{eq:fbb-arbitrary-1}~and~\eqref{eq:fbb-arbitrary-2} for DM
schedulability test under the first-fit strategy, where $\gamma$ is an
arbitrarily small positive number.
\end{theorem}
\begin{proof}
  This theorem is proved by a concrete input task system with $N=2M$
  tasks. There are $M$ light tasks with execution time $\frac{1}{3M}$
  and $M$ heavy tasks with execution time $\frac{1+\epsilon}{3}$, in
  which $\epsilon$ is a small positive real, i.e., $\epsilon > 0$. The
  $M$ light tasks are all with period $1-\delta$ and relative deadline
  $1-\delta$ with arbitrarily small and positive $\delta \ll \epsilon$. The $M$ heavy
  tasks are all with period $1$ and relative deadline $1$. Therefore,
  the $2M$ tasks are indexed such that
  \begin{itemize}
  \item $C_i=\frac{1}{3M}$, $T_i=D_i=1-\delta$, for $i=1,2,\ldots, M$, and
  \item $C_i=\frac{1+\epsilon}{3}$, $T_i=D_i=1$, for $i=M+1,M+2,\ldots, 2M$.
  \end{itemize}
  The setting of $0 < \delta \ll \epsilon$ is just to enforce the indexing. \emph{We will
  directly take $\delta \rightarrow 0$ for the rest of the proof.}\footnote{That is, the tightness examples in Theorems~\ref{theorem:sporadic-arbitrary-tight1},~\ref{theorem:sporadic-arbitrary-tight2},~and~\ref{theorem:sporadic-constrained-tight1} hold even when $\delta$ is $0$ if the tasks are indexed in the specified order. Setting $\delta$ to an arbitrarily small positive number is more precise mathematically, but this just complicates the presentation of the proofs without any added value. }

  By using Algorithm~\ref{alg:dmp} for the above task set when
  adopting
  Eqs.~\eqref{eq:fbb-arbitrary-1}~and~\eqref{eq:fbb-arbitrary-2} for
  DM schedulability test under the first-fit strategy, the $M$ light tasks
  are assigned on processor $1$. Then, when task $\tau_k$ with $k > M$
  is considered, the condition in Eq.~\eqref{eq:fbb-arbitrary-1}
  always fails for any of the first $k-M$ processors. Therefore, task
  $\tau_k$ is assigned to processor $k-M+1$, for $k=M+1, M+2, \ldots,
  2M-1$. It is then clear that task $\tau_{2M}$ cannot be assigned on
  any of the $M$ processors. Therefore, Algorithm~\ref{alg:dmp}
  returns ``no feasible solution is found''.

  By the above setting, we have $\sum_{i=1}^{2M} U_i =
  \frac{1+M+M\epsilon}{3}$. By using
  Lemma~\ref{lemma:lower-speed-bound}, we know that the speedup factor of the above task set
  is at least
  \begin{align*}
  \frac{1}{\frac{1+M+M\epsilon}{3}} = &3 - \frac{3+3\epsilon
    M}{M+1+\epsilon M} 
  = & 3-\frac{3}{M+1} - \gamma
  \end{align*}
  in which the factor $\gamma = \frac{3\epsilon M^2}{(M+1)(M+1+\epsilon M)}$
  becomes negligible when $\epsilon$ is sufficiently small.

  Such a factor can also be shown by a concrete
  partitioned schedule. By the pigeonhole principle, the
  solution that minimizes the maximum utilization of a processor is to
  assign a light task and a heavy task on a processor, in which the
  utilization on the processor is $\frac{1+\epsilon}{3}+\frac{1}{3M}$.
  Therefore, the task set is not schedulable on $M$ processors only
  when the speed is slower than $\frac{1+\epsilon}{3}+\frac{1}{3M}$,
  provided that $\delta$ is $0$.
  As a result, we reach the same conclusion.
\end{proof}

It may seem at first glance that the speedup factor $3-\frac{1}{M}$
in Corollary~\ref{cor:sporadic-arbitrary-arbitrary} is pessimistic,
since we do not actually use any property in the
exponential-time exact schedulability test. However, the
following theorem shows that the speedup factor $3-\frac{1}{M}$ is
asymptotically tight for an arbitrary fitting strategy, for any
schedulability tests used in Theorem~\ref{theorem:sporadic-arbitrary},
Corollary~\ref{cor:sporadic-arbitrary-poly}, and
Corollary~\ref{cor:sporadic-arbitrary-arbitrary}. As a result, to
improve the speedup factor, better fixed-priority scheduling
strategies or more precise analysis for concrete fitting strategies
are needed.

\begin{theorem}
\label{theorem:sporadic-arbitrary-tight2}
The speedup factor of Algorithm~\ref{alg:dmp} is at least
$3-\frac{3}{M+1}-\gamma$ under an arbitrary fitting strategy, for any
schedulability tests used in Theorem~\ref{theorem:sporadic-arbitrary},
Corollary~\ref{cor:sporadic-arbitrary-poly}, and
Corollary~\ref{cor:sporadic-arbitrary-arbitrary},  where $\gamma$ is an
arbitrarily small positive number. 
\end{theorem}
\begin{proof}
  This theorem is proved by a concrete input task system with $3M$
  tasks. Let $\delta$ and $\epsilon$ be very small positive real
  numbers, with $\delta \ll \epsilon$.  There are $M$ tasks with execution time $\frac{1}{3M}$,
  period $\infty$, and relative deadline $1-\delta$. There are $M$ tasks
  with execution time $\frac{\epsilon}{3}$, period $\epsilon$, and
  relative deadline $1$. There are $M$ tasks with execution time
  $\frac{1+\epsilon}{3}$, period $\infty$, and relative deadline
  $1+\delta$.

  Therefore, the $3M$ tasks are indexed such that
  \begin{itemize}
  \item $C_i=\frac{1}{3M}$, $T_i=\infty$, $D_i=1-\delta$, for $i=1,2,\ldots, M$, and
  \item $C_i=\frac{\epsilon}{3}$, $T_i=\epsilon$, $D_i=1$, for $i=M+1,M+2,\ldots, 2M$.
  \item $C_i=\frac{1+\epsilon}{3}$, $T_i=\infty$, $D_i=1+\delta$, for $i=2M+1,2M+2,\ldots, 3M$.
  \end{itemize}
  Again, the setting of $\delta \ll \epsilon$ is just to enforce the indexing. \emph{We
  will directly take $\delta \rightarrow 0$ for the rest of the proof.}
  Now, we consider a feasible task assignment for the first $3M-1$
  tasks, in which
  \begin{itemize}
  \item $\tau_1, \tau_2, \ldots, \tau_{M+1}$ are assigned on processor
    $1$, and
  \item $\tau_i$ and $\tau_{i+M-1}$ are assigned on processor $i-M$ for
    $i=M+2, M+3, \ldots, 2M$.
  \end{itemize}
  By using Algorithm~\ref{alg:dmp} for task $\tau_{3M}$, we know that
  task $\tau_{3M}$ cannot be feasibly assigned on any of the $M$
  processors since $\forall 0 < t \leq 1$ and $m=1,2,\ldots,M$, we
  have $\frac{1+\epsilon}{3} + \sum_{\tau_i \in {\bf T}_m}
  \ceiling{\frac{t}{T_i}}C_i > t$.  Therefore, Algorithm~\ref{alg:dmp}
  returns ``no feasible solution is found''.

  By the above setting, we know that
(1) $\sum_{i=1}^{3M} dbf(\tau_i, t)=0$ for $0 < t < 1$, (2)
$\sum_{i=1}^{3M} dbf(\tau_i, 1)=\frac{2\epsilon M+M+1}{3}$, and (3) for $1 < t$,
  \begin{align*}
\sum_{i=1}^{3M} dbf(\tau_i, t) \leq& \frac{1}{3}  +
  (t-1)\frac{1}{3}M + \frac{1+\epsilon}{3}M +
  \frac{\epsilon}{3}M\\
 =&
  \frac{2\epsilon M+Mt+1}{3}.
  \end{align*}
  As a result, $\max_{t > 0}\frac{\sum_{i=1}^{3M} dbf(\tau_i, t)}{M t}
  =  \frac{1+2\epsilon+\frac{1}{M}}{3}$, when $\epsilon$ is small enough.  
  Since $\max_{\tau_i\in {\bf T}} \Delta_i = \frac{1+\epsilon}{3}$ and
  $\sum_{i=1}^{3M}\frac{U_i}{M}=\frac{1}{3}$, by Lemma~\ref{lemma:lower-speed-bound}, the speedup
  factor of the above task set is
\begin{align*}
\frac{1}{\frac{1}{3M} + \frac{2\epsilon}{3} +
  \frac{1}{3}} 
= &\frac{3M}{M+2\epsilon M+1}\\  
=& 3 - \frac{3+6\epsilon
    M}{M+1+2\epsilon M}
  =  3-\frac{3}{M+1} - \gamma,
  \end{align*}
  in which the factor $\gamma=\frac{6\epsilon M^2}{(M+1)(M+1+2\epsilon M)}$
  becomes negligible when $\epsilon$ is sufficiently small. 
  The above analysis also works for the specific task partitioning which
  assigns three tasks $\tau_i$, $\tau_{i+M}$, and $\tau_{i+2M}$ on one
  processor for $i=1,2,\ldots,M$, and assigns the priority levels by
  using the optimal priority assignment by Audsley et
  al. \cite{audsley1993applying}, i.e., $\tau_{i+M}$ has the lowest priority on the processor.
\end{proof}

\section{Analysis for Constrained Deadlines}

This section presents the analysis for constrained-deadline sporadic
real-time systems. By Theorem~\ref{theorem:sporadic-arbitrary-tight1},
we know that the method by Fisher, Baruah, and Baker
\cite{DBLP:conf/ecrts/FisherBB06} leads to a speedup factor $3$ when
$M$ is sufficiently large even for implicit-deadline systems.  The
reason is mainly due to the pessimism of
Eq.~\eqref{eq:fbb-constrained-1} in the schedulability test. To get
better results, we do need better tests. A more precise strategy is to
simply use the exact test for constrained-deadline systems by spending
pseudo-polynomial time complexity. We have already shown 
(by Corollary~\ref{cor:sporadic-arbitrary-arbitrary} and Theorem~\ref{theorem:sporadic-arbitrary-tight2}) that spending more
time complexity does not help in arbitrary-deadline systems if the analysis does
not use the property of any specific fitting strategy. Is this
also the same for constrained-deadline systems?

We will first present the analysis by using TDA as the schedulability
test in Step~\ref{algo:fitting-test} in Algorithm~\ref{alg:dmp}. We
will conclude later that such high time complexity also does not help
reduce the speedup factor, compared to the results by using the hyperbolic bound in
Theorem~\ref{theorem:sporadic-general}. 

\subsection{Speedup Factor by Adopting TDA}
\label{sec:constrained-TDA}

Now, suppose that task $\tau_k$ is the first task that fails to be
assigned on any of the $M$ processors by using TDA schedulability
analysis in Step~\ref{algo:fitting-test} in Algorithm~\ref{alg:dmp}.
For notational brevity, let ${\bf T}^*$ be the set $\setof{\tau_1,
  \tau_2, \ldots, \tau_{k-1}}$ of the tasks

Therefore, we know that this leads to 
\[
\forall m, \forall  t,\mbox{ with } 0 < t \leq D_k, {\;\; \qquad \;\;} C_k +
\sum_{\tau_i \in {\bf T}_m} \ceiling{\frac{t}{T_i}}C_i > t.
\]
By taking a summation of all the $m=1,2,\ldots,M$ inequalities with
respect to any $t$, we know that the unschedulability of task $\tau_k$
by Algorithm~\ref{alg:dmp} implies that
\begin{equation}
  \label{eq:necessary-unschedulability-TDA-constrained}
  \forall  t \mbox{ with } 0 < t \leq D_k, {\;\; \qquad \;\;} MC_k +
\sum_{\tau_i \in  {\bf T}^*} \ceiling{\frac{t}{T_i}}C_i > Mt.
\end{equation}

Therefore, by taking the negation, we know that if 
\begin{equation}
  \label{eq:sufficient-schedulability-TDA-constrained}
  \exists  t \mbox{ with } 0 < t \leq D_k, {\;\; \mbox{and} \;\;} C_k +
\sum_{\tau_i \in  {\bf T}^*} \frac{\ceiling{\frac{t}{T_i}}C_i }{M}\leq t,
\end{equation}
then Algorithm \ref{alg:dmp} by using TDA should succeed to assign task $\tau_k$ on
one of the $M$ processors. This is basically very similar to TDA with
a minor difference by dividing the higher-priority workload by $M$.

Testing the schedulability condition of task $\tau_k$ according to
Eq.~\eqref{eq:sufficient-schedulability-TDA-constrained} can be done
by using the same strategy used in the ${\bf k^2U}$ framework
\cite{ChenHLRTSS2015} to prove
Theorem~\ref{theorem:sporadic-general} as follows.

 We classify the
$k-1$ tasks in ${\bf T}^*$ into two subsets.
\begin{itemize}
\item ${\bf T}^{*1}$ consists of the tasks in ${\bf T}^*$ with periods
  smaller than $D_k$.
\item ${\bf T}^{*2}$ consists of the tasks in ${\bf T}^*$ with periods 
  larger than or equal to $D_k$.
\end{itemize}

Now, let $C_k'$ be defined as follows:
\begin{equation}
  \label{eq:Ckprime}
  C_k' = C_k + \sum_{\tau_i \in  {\bf T}^{*2}} \frac{C_i }{M}.
\end{equation}
Suppose that we have $\kappa-1$ tasks in ${\bf T}^{*1}$. Clearly,
according to the definition $\kappa \geq 1$. Now, we can rewrite the 
condition in Eq.~\eqref{eq:sufficient-schedulability-TDA-constrained}
as follows: 
 if 
\begin{equation}
  \label{eq:sufficient-schedulability-TDA-k2u-constrained}
  \exists  t \mbox{ with } 0 < t \leq D_k {\;\; \mbox{and} \;\;} C_k' +
\sum_{\tau_i \in  {\bf T}^{*1}} \frac{\ceiling{\frac{t}{T_i}}C_i }{M}\leq t,
\end{equation}
then Algorithm \ref{alg:dmp} by using TDA should succeed to assign task $\tau_k$ on
one of the $M$ processors. 

For completeness, we repeat the definition of the \framework{}
framework and the key Lemma (with some simplifications to remove
individual coefficients for each task $\tau_i$) as follows.

\begin{definition}
  \label{def:kpoints}
  A $k$-point effective schedulability test is a sufficient schedulability test of a fixed-priority scheduling policy by verifying the existence of $t_j \in \setof{t_1, t_2, \ldots t_k}$ with $t_1 \leq t_2 \leq \cdots \leq t_k$ such that \begin{equation}
    \label{eq:precodition-schedulability}
    C_k + \sum_{i=1}^{k-1} \alpha t_i U_i + \sum_{i=1}^{j-1} \beta t_i U_i \leq t_j,
  \end{equation}
  where $C_k > 0$, $\alpha > 0$, $U_i > 0$, and $\beta >0$ are dependent upon the setting
  of the task models and task $\tau_i$.
\end{definition}

\begin{lemma}[Chen, Huang, and Liu \cite{ChenHLRTSS2015}]
\label{lemma:framework-constrained}
For a given $k$-point effective schedulability test, defined in
Definition~\ref{def:kpoints}, of a scheduling algorithm, in which $0 <
\alpha \neq \infty$, and $0 < \beta \neq \infty$, task $\tau_k$ is
schedulable by the scheduling algorithm if the following condition
holds
\begin{equation}
\label{eq:schedulability-constrained}
\frac{C_k}{t_k} \leq \frac{\frac{\alpha}{\beta}+1}{\prod_{j=1}^{k-1} (\beta U_j + 1)} - \frac{\alpha}{\beta}.
\end{equation}
\end{lemma}
\begin{proof}
  This comes from Lemma 1 in \cite{ChenHLRTSS2015}.
\end{proof}
By adopting the \framework{} framework
\cite{ChenHLRTSS2015}, we can conclude the
following theorem.
\begin{theorem}
  \label{thm:k2u-dm-hypobolic}
  If 
\[
\prod_{\tau_i \in {\bf T}^{*1}} (1+\frac{U_i}{M})\leq \frac{2}{ 1+\frac{C_k'}{D_k}} ,
\]
then task $\tau_k$ is schedulable under Algorithm~\ref{alg:dmp} by
using TDA.
\end{theorem}
\begin{proof}
  In the proof, we will reindex the tasks to satisfy the monotonicity
  of $t_i$ in Definition~\ref{def:kpoints}. That is, the $\kappa-1$
  higher-priority tasks in ${\bf T}^{*1}$ are reindexed to form the
  corresponding sequence $\tau_1, \tau_2, \ldots, \tau_{\kappa-1}$ for
  ensuring that the arrival times of the last jobs no later than $D_k$
  are in a non-decreasing order. For task $\tau_i$ in ${\bf T}^{*1}$,
  we set $t_i$ as $\floor{\frac{D_k}{T_i}}T_i$. Now, we reindex the
  $\kappa-1$ higher-priority tasks such that $t_1 \leq t_2 \leq \cdots
  \leq t_{\kappa-1}$. Moreover, let $t_{\kappa}$ be $D_k$.

  Since $T_i < D_k$ for any task $\tau_i$ in ${\bf T}^{*1}$, we have
  $t_i \geq T_i$.  Therefore, for a given $t_j$ with
  $j=1,2,\ldots,\kappa$, the demand requested up to time $t_j$ in
  Eq.~\eqref{eq:sufficient-schedulability-TDA-k2u-constrained} is at
  most
\begin{align}
  & C_k + \frac{\sum_{\tau_i \in {\bf T}^{*2}} C_i + \sum_{\tau_i \in {\bf T}^{*1}} \ceiling{\frac{t_j}{T_i}}C_i}{M}\nonumber\\
  = & C_k' +  \frac{\sum_{i=1}^{\kappa-1} \ceiling{\frac{t_j}{T_i}}C_i}{M}\nonumber\\
  \leq &  C_k' + \frac{\sum_{i=1}^{\kappa-1} \frac{t_i}{T_i}C_i}{M} +  \frac{\sum_{i=1}^{j-1} C_i}{M},\nonumber
\end{align}
where the inequality comes from the indexing policy defined above,
i.e., $\ceiling{\frac{t_j}{T_i}} \leq \frac{t_i}{T_i}+1$ if $j > i$
and $\ceiling{\frac{t_j}{T_i}} \leq \frac{t_i}{T_i}$ if $j \leq i$.

We only apply the test for these $\kappa$ different $t_i$ values, which is
equivalent to the test of the existence of $t_j$ for
$j=1,2,\ldots,\kappa$ such that $C_k' + \frac{\sum_{i=1}^{\kappa-1}
  \frac{t_i}{T_i}C_i}{M} + \frac{\sum_{i=1}^{j-1} C_i}{M} \leq
t_j$. This satisfies Definition~\ref{def:kpoints} (when $k$ is
$\kappa$) with $\alpha=\frac{1}{M}$ and
$\beta=\frac{1}{M}$. Therefore, by using
Lemma~\ref{lemma:framework-constrained}, if 
\[
\frac{C_k'}{D_k} \leq \frac{1+1}{\prod_{\tau_i \in {\bf T}^{*1}} (\frac{U_i}{M} + 1)} - 1,
\]
then task $\tau_k$ is schedulable under Algorithm~\ref{alg:dmp} by
using TDA. By reorganizing the above equation, we reach the
conclusion.
\end{proof}

The following corollary comes from the same proof of
Lemma~\ref{lemma:framework-constrained} and
Theorem~\ref{thm:k2u-dm-hypobolic}, which shows that the
schedulability condition also implies a lower bound of the workload 
$\frac{C_k'}{D_k}+ \frac{ \sum_{\tau_i \in {\bf T}^{*1}} t_i U_i}{M
  D_k}$.
\begin{corollary}
  \label{cor:k2u-dm-extreme}
  If the schedulability condition in
  Theorem~\ref{thm:k2u-dm-hypobolic} is violated, i.e., $\prod_{\tau_i
    \in {\bf T}^{*1}} (1+\frac{U_i}{M})> \frac{2}{
    1+\frac{C_k'}{D_k}}$, then 
\[
\frac{C_k'}{D_k}+ \frac{ \sum_{\tau_i \in {\bf T}^{*1}} \floor{\frac{D_k}{T_i}}T_i U_i}{M
  D_k} > \prod_{\tau_i \in {\bf T}^{*1}} \left(1+\frac{U_i}{M}\right).
\]
\end{corollary}
\begin{proof}
  This comes from the same proof as in
  Lemma~\ref{lemma:framework-constrained} (Lemma 1 in
  \cite{ChenHLRTSS2015}) by changing the objective
  from minimizing $C_k^*$ to minimizing $C_k^* +
  \sum_{i=1}^{k-1}\alpha t_i U_i$ to enforce the
  unschedulability. This property has been provided by
  Chen, Huang, and Liu
  \cite{ChenHLRTSS2015} to prove the speedup factor of
  deadline-monotonic scheduling in uniprocessor systems. Moreover,
  this condition has also been exploited by
  von der Br\"uggen, Chen, and Huang
  \cite{DBLP:conf/ecrts/BruggenCH15} to prove the speedup factor of
  non-preemptive DM scheduling with respect to non-preemptive EDF
  scheduling. 
\end{proof}

The remaining proofs in this section require some mathematical tools,
which are provided in the following lemmas.

\begin{lemma}
  \label{lemma:hypo-worst-case}
  Suppose that $\sum_{\tau_i \in {\bf T}^{1*}} U_i > 0$ is fixed and
  $U_i \geq 0$ for each task $\tau_i$. Then, $\prod_{\tau_i \in {\bf
      T}^{1*}} \left(\frac{U_i}{M} + 1\right)$ is maximized when $U_1=U_2=
  \cdots = U_{|{\bf T}^{1*}|}$.
\end{lemma}
\begin{proof}
  This can be easily proved by the fact that $\prod_{\tau_i \in {\bf
      T}^{1*}} \left(\frac{U_i}{M} + 1\right)$ is a concave function
  with respect to the (non-negative) values of  $U_i$.
\end{proof}

\begin{lemma}
  \label{lemma:hypo-worst-case-total}
  The infimum $\sum_{\tau_i \in {\bf T}^{1*}} \frac{U_i}{M}$ to
  enforce $\prod_{\tau_i \in {\bf T}^{1*}} \left(\frac{U_i}{M} + 1\right) >
  x$ is $\ln(x)$.
\end{lemma}
\begin{proof}
  This can be derived by using Lagrange Multiplier Method to find the
  minimum $\sum_{\tau_i \in {\bf T}^{1*}} \frac{U_i}{M}$ such that
  $\prod_{\tau_i \in {\bf T}^{1*}} \left(\frac{U_i}{M} + 1\right) \geq
  x$. By Lemma~\ref{lemma:hypo-worst-case}, it is clear that the worst
  case is to have all the tasks with the same utilization. Suppose
  that $|{\bf T}^{1*}|$ is $n$. We know that all the tasks are with
  utilization $M(x^{\frac{1}{n}}-1)$, and the utilization bound
  $\sum_{\tau_i \in {\bf T}^{1*}} \frac{U_i}{M}$ is
  $n(x^{\frac{1}{n}}-1)$. This converges to $\ln(x)$ when $n$
  approaches $\infty$.
\end{proof}

With the above discussions, we can conclude the speedup factor.

\begin{theorem}
  \label{thm:speedup-dm}
  The speedup factor of Algorithm~\ref{alg:dmp} for
  constrained-deadline task systems by using TDA is
  $\frac{1}{W(0.5)}\approx 2.84306$, where $W(z)$ is the Lambert W
  function, i.e., the unique solution of $z=W(z)e^{W(z)}$.
\end{theorem}
\begin{proof}
  If $\prod_{\tau_i \in {\bf T}^{*1}}
  (\frac{U_i}{M} + 1) \geq 2$, we can already conclude that
  $\sum_{\ \tau_i \in {\bf T}^{*1}} \frac{U_i}{M} \geq
  \ln{2}$ by using Lemma~\ref{lemma:hypo-worst-case-total}, and the speedup factor is $1/\ln{2} < 2.84306$ for such a
  case. We focus on the other case with $\prod_{\tau_i \in
    {\bf T}^{*1}} (\frac{U_i}{M} + 1) < 2$.  Suppose that $\sigma$ is
  $\frac{2}{\prod_{\tau_i \in {\bf T}^{*1}} (\frac{U_i}{M} + 1)}-1$,
  in which $\sigma > 0$.

If $\tau_k$ is not schedulable under Algorithm
\ref{alg:dmp} by TDA, then
{\small  \begin{align*}
&\frac{C_k}{D_k} + \frac{\sum_{i=1}^{k-1} dbf(\tau_i,D_k)}{M D_k}\\
= \;\;&\frac{C_k}{D_k} + \frac{\sum_{\tau_i \in {\bf T}^{*2}}C_i}{M D_k}+ \frac{\sum_{\tau_i \in {\bf T}^{*1}}
      dbf(\tau_i,D_k)}{M D_k} \\
\geq\;\;& \frac{C_k'}{D_k} + \frac{
      \sum_{\tau_i \in {\bf T}^{*1}} \floor{\frac{D_k}{T_i}}T_i  U_i}{M D_k} \\
    >_1
    &\prod_{\tau_i \in {\bf T}^{*1}} \left(1+\frac{U_i}{M}\right) = \frac{1+\sigma}{2},
  \end{align*}}
  where $>_1$ is by Corollary~\ref{cor:k2u-dm-extreme}.
  Suppose that $\frac{C_k}{D_k}$ is $x$. Therefore, we know that 
  \begin{equation}
    \label{eq:DM-dbf-final}
\sum_{\tau_i \in {\bf T}^*} \frac{dbf(\tau_i,D_k)}{M} > \frac{1+\sigma}{2}-x.    
  \end{equation}

  Moreover, with $\prod_{\tau_i \in {\bf T}^{*1}} (1+\frac{U_i}{M}) =
  \frac{2}{ 1+\sigma}$ and the fact that $\tau_k$ is not schedulable under Algorithm
\ref{alg:dmp} by using TDA, by
Lemma~\ref{lemma:hypo-worst-case-total}, we have
  \begin{equation}
    \label{eq:DM-utilization}
    \sum_{\tau_i \in {\bf T}^{*1}} \frac{U_i}{M} > \ln(\frac{2}{1+\sigma}).
  \end{equation}

  For the rest of the proof, we consider two separate
  cases:\footnote{These cases are concluded by the following trick:
    Suppose that $g_1(y)$ is an increasing function and $g_2(y)$ is a
    decreasing function with respect to $y$ when $y \geq 0$. Then, we
    know that the lower bound on $\max\{g_1(y), g_2(y)\}$ is the
    intersection of these two functions, i.e., $g_1(y^*)$, where
    $g_1(y^*)=g_2(y^*)$ if such a value $y^*$ exists. That is,
    $\max\{g_1(y), g_2(y)\} \geq \min_{y \geq 0}\max\{g_1(y),
    g_2(y)\}= g_1(y^*)$. }

  {\bf Case 1} $x \geq \sigma$: This is an easier case. We can conclude the
  speedup factor by using Eq.~\eqref{eq:DM-utilization}
  \begin{align*}
 &\max\left\{\frac{C_k}{D_k}, \sum_{\tau_i \in {\bf T}^*} \frac{U_i}{M}\right\} \geq
\max\left\{\sigma, \ln(\frac{2}{1+\sigma})\right\} \\
\geq &\min_{\sigma > 0}\left\{\max\{\sigma, \ln(\frac{2}{1+\sigma})\} \right\}=_1 \frac{1}{2.66793},
  \end{align*}
  where $=_1$ holds when $e^{\sigma}(1+\sigma)=2$.

  {\bf Case 2} $x <\sigma$: 
 There are two subcases
  \begin{itemize}
  \item {\bf Case 2a}: If $x > \frac{1+\sigma}{4}$, we know that $x >
    \frac{1+\sigma}{2}-x$. Therefore,  by Eq.~\eqref{eq:DM-dbf-final},
    $\max\left\{\frac{C_k}{D_k},
      \sum_{\tau_i \in {\bf T}^*} \frac{dbf(\tau_i,D_k)}{M}\right\} \geq x
    > \frac{1+\sigma}{4}$. Hence,
    \begin{align*}
 &\max\left\{\frac{C_k}{D_k}, \sum_{\tau_i \in {\bf T}^*} \frac{U_i}{M},
    \sum_{\tau_i \in {\bf T}^*} \frac{dbf(\tau_i,D_k)}{M}\right\}\\
 > &\max\left\{\frac{1+\sigma}{4}, \ln(\frac{2}{1+\sigma})\right\} 
      \geq  \min_{y \geq 0} \max\left\{y,  \ln\frac{1}{2y}\right\} \\
      =& W(0.5)\approx \frac{1}{2.84306},
    \end{align*}
    where $=$ holds when $ye^{y}=0.5$.

 \item {\bf Case 2b}:  If $x \leq \frac{1+\sigma}{4}$, we know that $x \leq
    \frac{1+\sigma}{2}-x$. Therefore,  by Eq.~\eqref{eq:DM-dbf-final},
    $\max\left\{\frac{C_k}{D_k},
      \sum_{\tau_i \in {\bf T}^*} \frac{dbf(\tau_i,D_k)}{M}\right\} 
    > \frac{1+\sigma}{2}-x \geq \frac{1+\sigma}{4}$. Hence,
    \begin{align*}
&\max\left\{\frac{C_k}{D_k}, \sum_{\tau_i \in {\bf T}^*} \frac{U_i}{M},
    \sum_{\tau_i \in {\bf T}^*} \frac{dbf(\tau_i,D_k)}{M}\right\}\\
 >       &\max\left\{\frac{1+\sigma}{4} ,  \ln\left(\frac{2}{1+\sigma}\right)\right\} 
      \geq  \min_{y \geq 0} \max\left\{y ,  \ln(\frac{1}{2y})\right\}\\
      =& W(0.5)\approx \frac{1}{2.84306},
    \end{align*}
    where $=$ holds when $ye^{y}=0.5$.
  \end{itemize}

  Therefore, by all the above cases, we know that
{\small \begin{equation*}
\max\left\{\frac{C_k}{D_k}, \sum_{\tau_i \in {\bf T}^*} \frac{U_i}{M},
    \sum_{\tau_i \in {\bf T}^*} \frac{dbf(\tau_i,D_k)}{M}\right\} > W(0.5)\approx \frac{1}{2.84306},    
  \end{equation*}}which concludes the proof by applying Lemma~\ref{lemma:lower-speed-bound}.
\end{proof}

\subsection{Speedup Factor by Hyperbolic Bound}
\label{sec:constrained-hyperbolic}

This subsection further presents the speedup factor of
Algorithm~\ref{alg:dmp} when adopting the hyperbolic bound in
Eq.~\eqref{eq:schedulability-sporadic-any-a} for testing the
schedulability of DM scheduling. The speedup factor analysis in
Theorem~\ref{thm:speedup-dm} for TDA schedulability analysis relies
only on the violation of the schedulability condition in
Theorem~\ref{thm:k2u-dm-hypobolic}. We will show that adopting the
hyperbolic bound of Eq.~\eqref{eq:schedulability-sporadic-any-a}
results in the same condition in
Theorem~\ref{thm:k2u-dm-hypobolic}. Therefore, we can reach the same
conclusion as in Theorem~\ref{thm:speedup-dm} by using the hyperbolic
bound in polynomial time. We use the same notations, e.g., ${\bf
  T}^{*1}$, ${\bf T}^{*2}$, $\kappa$, etc., as used in
Section~\ref{sec:constrained-TDA}.

\begin{theorem}
  \label{thm:k2u-dm-hypobolic-poly}
  If 
\[
\prod_{\tau_i \in {\bf T}^{*1}} \left(1+\frac{U_i}{M}\right)\leq \frac{2}{ 1+\frac{C_k'}{D_k}} ,
\]
then task $\tau_k$ is schedulable under Algorithm~\ref{alg:dmp} by
using the hyperbolic bound in
Eq.~\eqref{eq:schedulability-sporadic-any-a}, where  $C_k'$ is $C_k+ \frac{\sum_{\tau_i \in {\bf T}^{*2}} C_i}{M}$.
\end{theorem}
\begin{proof}
  We prove this by contrapositive. Suppose that task $\tau_k$ is not
  schedulable under Algorithm~\ref{alg:dmp} by using the hyperbolic
  bound in Eq.~\eqref{eq:schedulability-sporadic-any-a} . Therefore,
 for $m=1,2,\ldots, M$, we have 
\[
\left(\frac{C_k + \sum_{\tau_i \in {\bf T}_m^2} C_i}{D_k}+1\right) \prod_{\tau_i \in {\bf T}_m^1}
(U_i + 1) > 2.
\]
  By multiplying the above $M$ inequalities we reach
  \begin{align*}
2^M <& \left(\prod_{m=1}^{M}\left(\frac{C_k + \sum_{\tau_i \in {\bf T}_m^2} C_i}{D_k}+1\right)\right)  \left(\prod_{\tau_i \in {\bf T}^{*1}}
(U_i + 1)\right)\\
    \leq_1 & \left(\frac{C_k + \frac{\sum_{\tau_i \in {\bf T}^{*2}} C_i}{M}}{D_k}+1\right)^M   \left(\prod_{\tau_i \in {\bf T}^{*1}}
\left(\frac{U_i}{M} + 1\right)^M\right),
  \end{align*}
  where $\leq_1$ comes from Lemma~\ref{lemma:hypo-worst-case} (for the
  first part) and from the fact $(1+U_i) \leq (1+U_i/M)^M$ when $M$ 
  is a positive integer and $U_i \geq 0$
  (for the second part).
  Therefore, we conclude that the unschedulability of task $\tau_k$
  implies that
\[
2 < \left(\frac{C_k'}{D_k}+1\right)   \left(\prod_{\tau_i \in {\bf T}^{*1}}
\left(\frac{U_i}{M} + 1\right)\right).
\]
  By contrapositive, we reach the conclusion.
\end{proof}

\begin{theorem}
  \label{thm:speedup-dm-hyper}
  The speedup factor of Algorithm~\ref{alg:dmp} by using the
  hyperbolic bound in Eq.~\eqref{eq:schedulability-sporadic-any-a} is
  $\frac{1}{W(0.5)}\approx 2.84306$, where $W(z)$ is the Lambert W
  function, i.e., $z=W(z)e^{W(z)}$.
\end{theorem}
\begin{proof}
  Since the schedulability condition remains the same as in
  Theorem~\ref{thm:k2u-dm-hypobolic}, the speedup factor is also the
  same as in Theorem~\ref{thm:speedup-dm}.
\end{proof}

\subsection{Tightness Analysis}
\label{sec:constrained-tightness}

We conclude this section by showing that the above speedup factor
analysis  is tight when $M$ is sufficiently large under an arbitrary
fitting strategy. 
\begin{theorem}
  \label{theorem:sporadic-constrained-tight1} For constrained-deadline
  task systems, the speedup factor of Algorithm~\ref{alg:dmp} is at
  least $\frac{1}{W(0.5)} \approx 2.84306$ when adopting TDA or
  the hyperbolic bound in Eq~\eqref{eq:schedulability-sporadic-any-a}
  for DM schedulability test under an arbitrary fitting strategy, when
  $M$ is sufficiently large.
\end{theorem}
\begin{proof}
  We prove this theorem by providing a concrete task system by
  assuming that $M$ is sufficiently large. There are $N=2M^2+ 1$
  tasks. Let $f$ be $\frac{2}{\frac{1}{W(0.5)}} \approx
  0.7034674$, i.e., $\ln(\frac{1}{f}) = \frac{f}{2}$. Let $\delta$ be an arbitrarily small  positive number just for
  enforcing the indexing:
  \begin{itemize}
 \item $T_i = D_i = f + (\ceiling{\frac{i}{M}}-1)\frac{1-f}{M-1}$, $C_i =
    \frac{1-f}{M-1}$, for $i=1,2,\ldots, M^2$,
  \item $T_i=\infty$, $D_i=1+\delta$, $C_i=\frac{1.5f-1}{M}\approx \frac{0.0552}{M}$, for
    $i=M^2+1,M^2+2,\ldots, 2M^2$,  and
  \item $T_N = \infty$, $D_N = 1+2\delta$, $C_N = 0.5f+\epsilon$, with
    $N=2M^2+1$, where $\epsilon$ is a positive small number.
  \end{itemize}
  \emph{We again simply take $\delta$ to $0$ for the rest of the proof.}
  Moreover, $\frac{1}{M}$ is also consider negligible for the
  simplicity of computation, since $M$ is assumed to be sufficiently
  large.

  For an arbitrary fitting algorithm, consider the following task
  assignment by assigning task $\tau_{i+jM}$ to processor $i$ with $i=1,2,\ldots,M$
  for every $j=0,1,2,\ldots,2M-1$. It is not difficult to see that the
  above task assignment can be achieved feasibly and results in a
  feasible task assignment for the first $2M^2$ tasks. The set of the
  tasks assigned on processor $m$ is denoted as ${\bf T}_m$. Now,
  consider task $\tau_{2M^2+1}$, i.e., $\tau_N$ to be assigned on
  processor $m$. The overall execution time request at time $0$ is $0.5f+\epsilon +
  1.5f-1+1-f=f+\epsilon$ on processor $m$. Therefore, it can be easily seen that task
  ${\bf T}_m\cup \setof{\tau_N}$ is not schedulable under DM scheduling on processor $m$ since the TDA test
  in Eq.~\eqref{eq:fbb-constrained-1} fails. As a result, task
  $\tau_N$ cannot be assigned on any processor.


  In this input task set, the utilization of
  the individual task and $\frac{C_i}{D_i}$ are not more than $\frac{f}{2}+\epsilon$ for each
  task $\tau_i$. Moreover, the overall task utilization is
  $M\sum_{i=0}^{M-1}\frac{\mu}{f+i \mu}$, where $\mu =
  \frac{1-f}{M-1}$. Due to the assumption that $M$ is sufficiently
  large, the above total utilization is a \emph{left Riemann sum},
  i.e., $M\sum_{i=0}^{M-1}\frac{\mu}{f+i \mu} \approx M\int_{0}^{1-f}
  \frac{1}{f+ x} dx= M\ln(\frac{1}{f})$. By the fact that
  $\ln(\frac{1}{f}) = \frac{f}{2}$, we know that the total utilization
  is $M\frac{f}{2}$ when $M$ is sufficiently large.  

  Now, we examine $\max_{t > 0}\frac{\sum_{i=1}^{N} dbf(\tau_i,
    t)}{Mt}$. By definition, $\max_{t > 0}\frac{\sum_{i=1}^{N}
    dbf(\tau_i, t)}{Mt}> \max_{t > 0}\frac{\sum_{i=1}^{N-1}
    dbf(\tau_i, t)}{Mt}$. By the construction of the task set, we have
  (1) $\frac{\sum_{i=1}^{N-1} dbf(\tau_i, t)}{Mt} = 0$ if $0 < t < 1$,
  and (2) $\frac{\sum_{i=1}^{N-1} dbf(\tau_i, 1)}{M} = \frac{f}{2}$ if
  $t=1$. Therefore, we know that the speedup factor for this task set
  is purely dominated by $\max_{t > 0}\frac{\sum_{i=1}^{N-1}
    dbf(\tau_i, t)}{Mt}$.  

  However, proving that
  $\frac{\sum_{i=1}^{N-1} dbf(\tau_i, t)}{Mt} \leq \frac{f}{2}$ if $t
  > 1$ is pretty complicated.  The proof involves quite some
  mathematical derivations, and is left in Appendix A\ifbool{techreport}{.}{ in \cite{DBLP:journals/corr/Chen15k}.}
It should be clear that we can ignore the
  $M-1$ duplicated copies of the tasks by considering only the tasks
  assigned on one processor (before considering $\tau_N$). 
The complete proof is Appendix A\ifbool{techreport}{.}{ in \cite{DBLP:journals/corr/Chen15k}.}
The proof strategy first
  makes an over-approximation, denoted as $dbf^{\sharp}(t)$, of the
  sum $\sum_{\tau_i \in {\bf T}_m} dbf(\tau_i, t)$ of the demand bound
  functions at time $t$ (after removing the $M-1$ duplicated copies). Based on such an over-approximation, it can
  be shown that $\max_{t>0} \frac{dbf^{\sharp}(t)}{t}$ happens when
  $t$ is a positive integer for ${\bf T}_m$. It can then be proved that the
  maximum $\frac{dbf^{\sharp}(t)}{t}$ happens when $t=1$ or
  $t=\infty$, in which both leads $\frac{dbf^{\sharp}(t)}{t}$ to
  $\frac{f}{2}$. 
  
  Figure~\ref{fig:segment-curve} draws the functions
 $dbf^{\sharp}(t)$ and $\frac{dbf^{\sharp}(t)}{t}$ in
the range of $[1, 6]$.
 Figure~\ref{fig:dbf-tight} provides an illustrative
  view of $\frac{dbf^{\sharp}(t)}{t}$ from $t=1,2,\ldots, 4000$.
  It can also be easily shown by assigning $\tau_N$ to one processor
  alone, we can find a task partition that requires a speedup factor
  asymptotically equal to $\frac{2}{f}$ when $M$ is sufficiently large.
\end{proof}


\begin{figure}[t]
\centering
\subfloat[$dbf^{\sharp}(t)$ ]{\label{fig:segment-curve-A}
  \includegraphics[width=0.75\columnwidth]{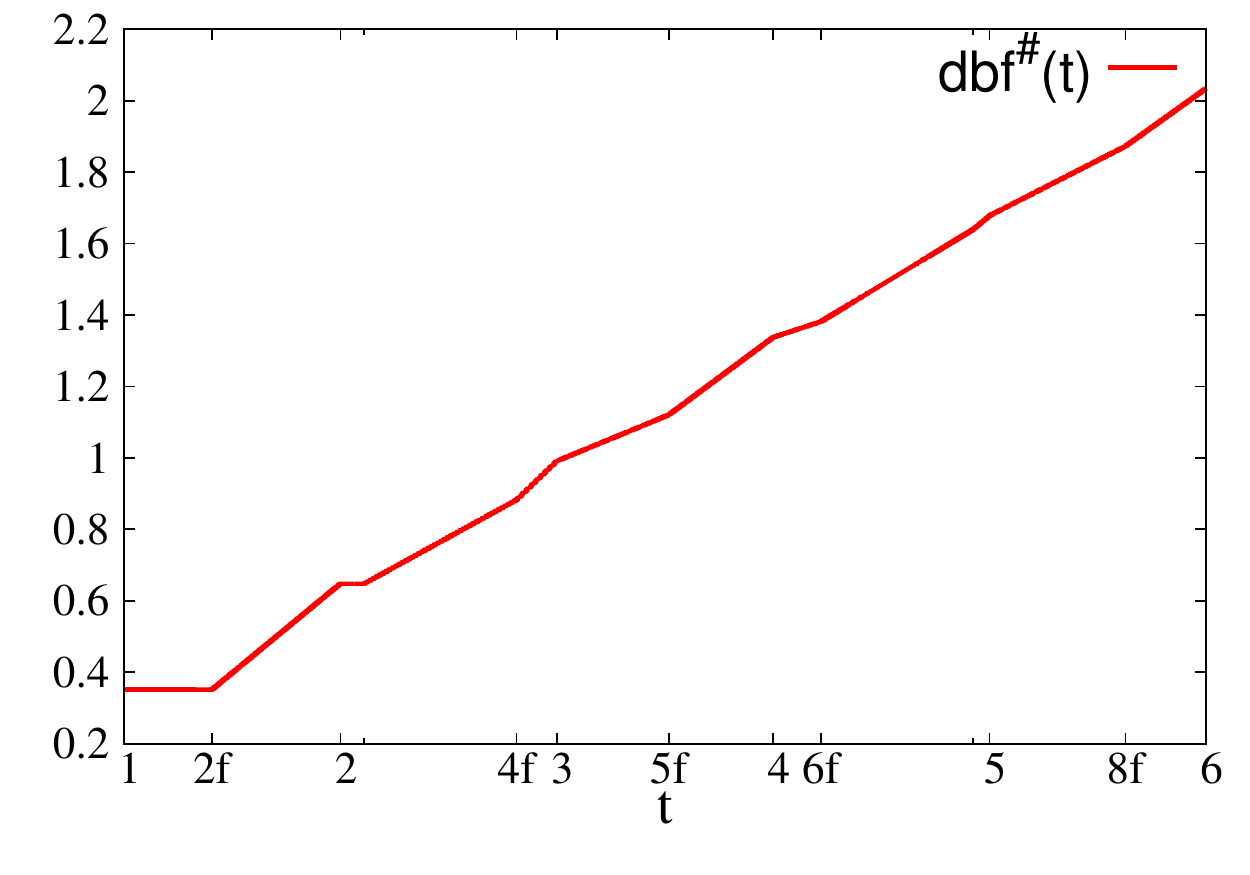}}\\
\subfloat[$dbf^{\sharp}(t)/t$]{\label{fig:segment-curve-B}
  \includegraphics[width=0.75\columnwidth]{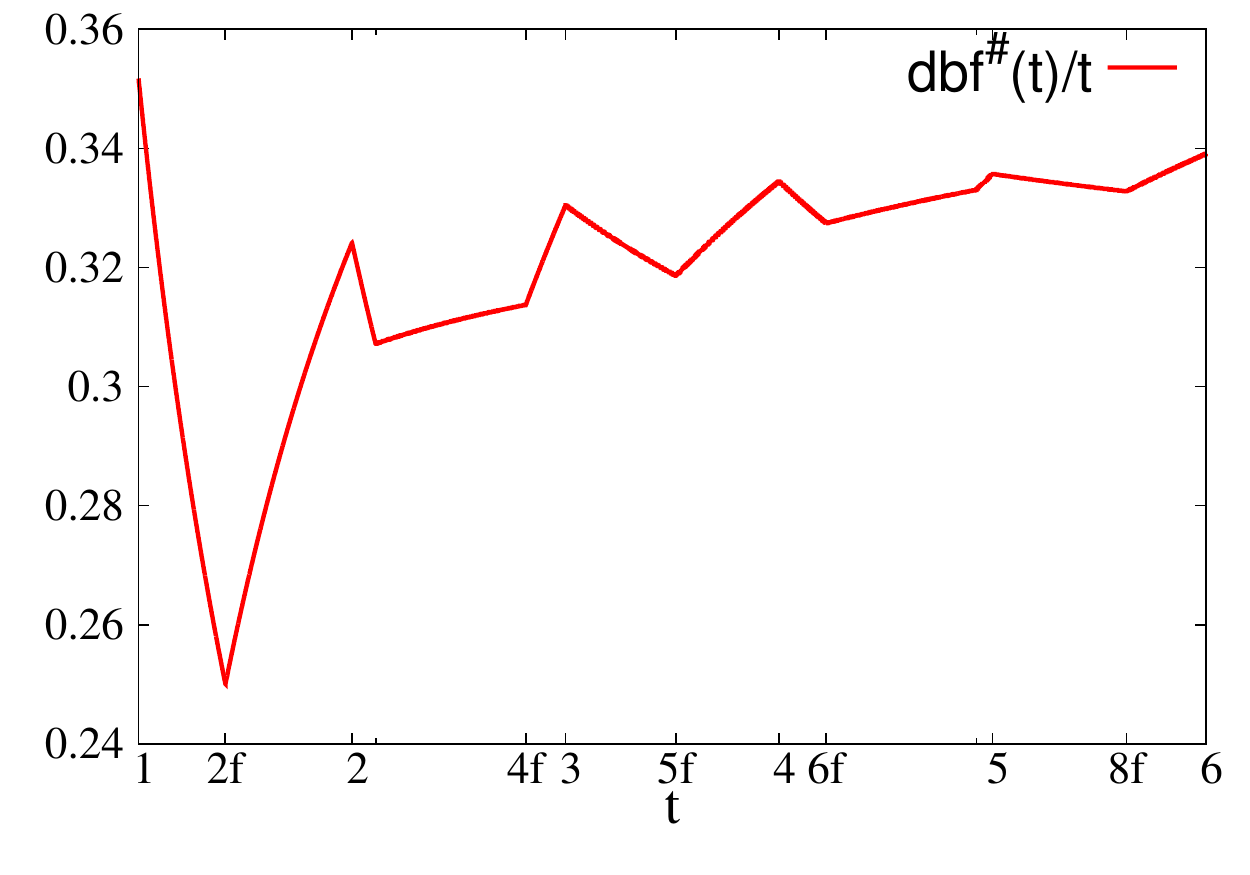}}
  \caption{Functions $dbf^{\sharp}(t)$ and $\frac{dbf^{\sharp}(t)}{t}$
    for the proof of Theorem~\ref{theorem:sporadic-constrained-tight1}
    with $t$ in
the range of $[1, 6]$.}
  \label{fig:segment-curve}
\end{figure}

 \begin{figure}[t]
\centering
\subfloat[$ 1 \leq t \leq 30$ ]{
  \includegraphics[width=0.75\columnwidth]{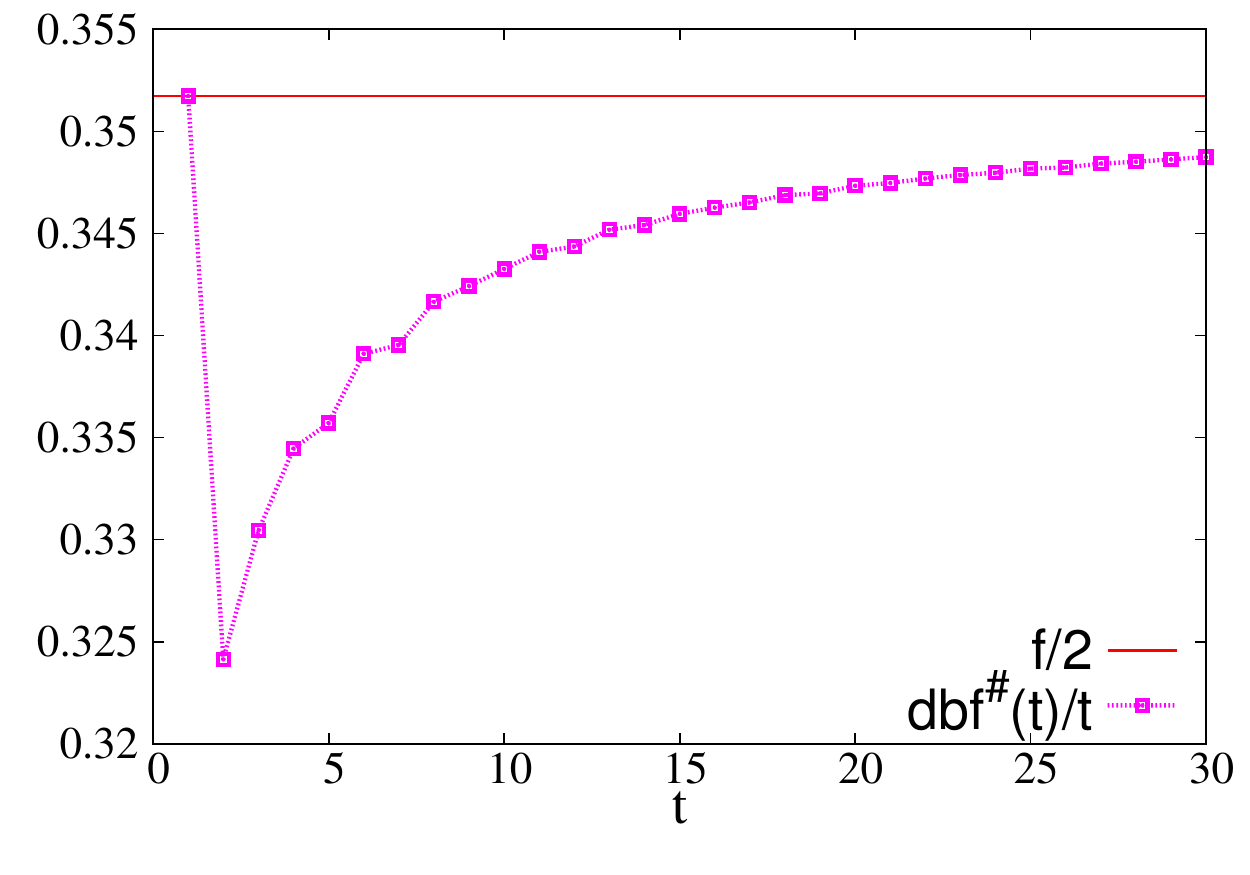}}\\
\subfloat[$ 30 \leq t \leq 4000$]{
  \includegraphics[width=0.75\columnwidth]{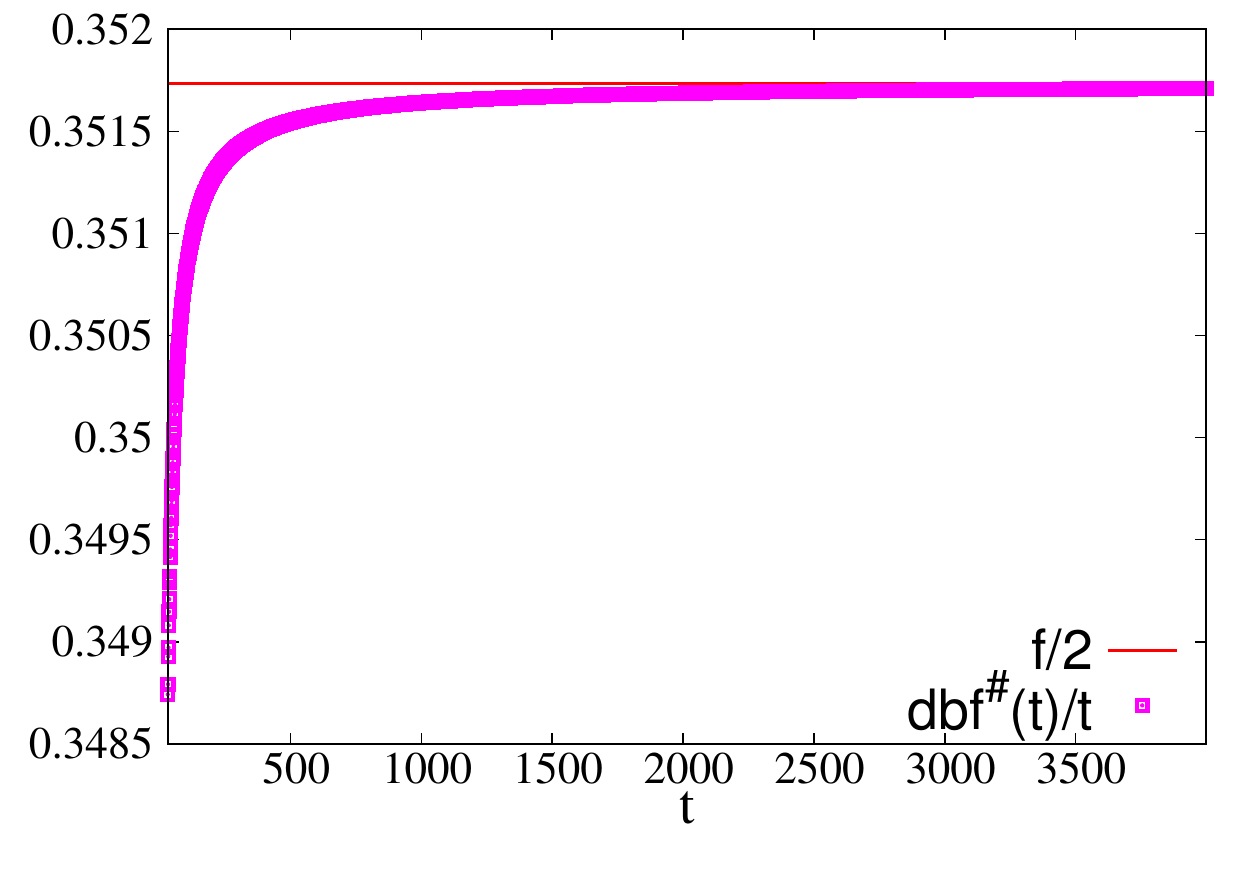}}
  \caption{$\frac{dbf^{\sharp}(t)}{t}$ when $t$ is a positive integer
    number (the curve is just for visualization)  for the proof of Theorem~\ref{theorem:sporadic-constrained-tight1}, and the reference point $\frac{f}{2} \approx 0.3517337$}
  \label{fig:dbf-tight}
\end{figure}

\section{Concluding Remarks}
\label{sec:conclusion}

This paper provides detailed analysis for  the deadline-monotonic partitioning algorithm
proposed by Fisher, Baruah, and Baker
\cite{DBLP:conf/ecrts/FisherBB06} for multiprocessor partitioned fixed-priority scheduling, by
using exact schedulability tests and approximated schedulability
tests. It may seem at first glance that using exact schedulability
tests is more precise, but the proof shows that such exact tests are
with the same speedup factors as approximated tests.  We show that the
deadline-monotonic partitioning algorithm has a speedup factor
$3-\frac{1}{M}$ when considering task systems with arbitrary
deadlines. Such a factor holds
for polynomial-time schedulability tests and exponential-time
schedulability tests.  Moreover, we also improve the speedup factor to
$2.84306$ when considering constrained-deadline task systems.

The speedup factor analyses in this paper are asymptotically tight
if the analysis is oblivious of the particular fitting strategy used.
In all the tightness analyses,
we only take Lemma~\ref{lemma:lower-speed-bound}, which also
implicitly implies the
reference to optimal global scheduling.  The tasks are designed on
purpose, e.g., $M^2$ tasks (instead of $M$ tasks) with period $\infty$
in the proof of Theorem~\ref{theorem:sporadic-constrained-tight1}, to
show that such factors also hold (asymptotically or with minor
changes) for optimal partitioned scheduling.  However, this does not
limit the potential to have better speedup factors by adopting better
fixed-priority scheduling strategies or more precise analysis for
concrete fitting strategies.

\noindent{\bf Acknowledgements.} This paper is supported by DFG, as part of the Collaborative Research Center SFB876 (http://sfb876.tu-dortmund.de/).

\begin{spacing}{1}
{\footnotesize
\def\IEEEbibitemsep{0pt}
  \bibliographystyle{abbrv}
\bibliography{real-time}

\begin{thebibliography}{10}

\bibitem{DBLP:conf/ecrts/AlbersS04}
K.~Albers and F.~Slomka.
\newblock An event stream driven approximation for the analysis of real-time
  systems.
\newblock In {\em ECRTS}, pages 187--195, 2004.

\bibitem{audsley1993applying}
N.~Audsley, A.~Burns, M.~Richardson, K.~Tindell, and A.~J. Wellings.
\newblock Applying new scheduling theory to static priority pre-emptive
  scheduling.
\newblock {\em Software Engineering Journal}, 8(5):284--292, 1993.

\bibitem{BaruahRTSS2011}
S.~Baruah.
\newblock The partitioned {EDF} scheduling of sporadic task systems.
\newblock In {\em Real-Time Systems Symposium (RTSS)}, pages 116 --125, 2011.

\bibitem{DBLP:conf/rtss/BaruahF05}
S.~K. Baruah and N.~Fisher.
\newblock The partitioned multiprocessor scheduling of sporadic task systems.
\newblock In {\em RTSS}, pages 321--329, 2005.

\bibitem{DBLP:journals/tc/BaruahF06}
S.~K. Baruah and N.~Fisher.
\newblock The partitioned multiprocessor scheduling of deadline-constrained
  sporadic task systems.
\newblock {\em IEEE Trans. Computers}, 55(7):918--923, 2006.

\bibitem{DBLP:conf/rtss/BaruahMR90}
S.~K. Baruah, A.~K. Mok, and L.~E. Rosier.
\newblock Preemptively scheduling hard-real-time sporadic tasks on one
  processor.
\newblock In {\em IEEE Real-Time Systems Symposium}, pages 182--190, 1990.

\bibitem{DBLP:journals/tc/Bini15}
E.~Bini.
\newblock The quadratic utilization upper bound for arbitrary deadline
  real-time tasks.
\newblock {\em {IEEE} Trans. Computers}, 64(2):593--599, 2015.

\bibitem{bini2003rate}
E.~Bini, G.~C. Buttazzo, and G.~M. Buttazzo.
\newblock Rate monotonic analysis: the hyperbolic bound.
\newblock {\em Computers, IEEE Transactions on}, 52(7):933--942, 2003.

\bibitem{DBLP:journals/tc/BiniNRB09}
E.~Bini, T.~H.~C. Nguyen, P.~Richard, and S.~K. Baruah.
\newblock A response-time bound in fixed-priority scheduling with arbitrary
  deadlines.
\newblock {\em {IEEE} Trans. Computers}, 58(2):279--286, 2009.

\bibitem{DBLP:dblp_journals/tc/BurchardLOS95}
A.~Burchard, J.~Liebeherr, Y.~Oh, and S.~H. Son.
\newblock New strategies for assigning real-time tasks to multiprocessor
  systems.
\newblock pages 1429--1442, 1995.

\bibitem{DBLP:conf/rtss/ChakrabortyKT02}
S.~Chakraborty, S.~K{\"u}nzli, and L.~Thiele.
\newblock Approximate schedulability analysis.
\newblock In {\em IEEE Real-Time Systems Symposium}, pages 159--168, 2002.

\bibitem{DBLP:journals/corr/Chen15k}
J.~Chen.
\newblock Partitioned multiprocessor fixed-priority scheduling of sporadic
  real-time tasks.
\newblock {\em Computing Research Repository (CoRR)}, abs/1505.04693,
  \url{http://arxiv.org/abs/1505.04693}, 2015.

\bibitem{Chakraborty2011a}
J.-J. Chen and S.~Chakraborty.
\newblock Resource augmentation bounds for approximate demand bound functions.
\newblock In {\em IEEE Real-Time Systems Symposium}, pages 272 -- 281, 2011.

\bibitem{ChenECRTS12}
J.-J. Chen and S.~Chakraborty.
\newblock Partitioned packing and scheduling for sporadic real-time tasks in
  identical multiprocessor systems.
\newblock In {\em ECRTS}, pages 24--33, 2012.

\bibitem{DBLP:journals/rts/ChenC13}
J.-J. Chen and S.~Chakraborty.
\newblock Resource augmentation for uniprocessor and multiprocessor partitioned
  scheduling of sporadic real-time tasks.
\newblock {\em Real-Time Systems}, 49(4):475--516, 2013.

\bibitem{ChenHLRTSS2015}
J.-J. Chen, W.-H. Huang, and C.~Liu.
\newblock {k2U}: A general framework from k-point effective schedulability
  analysis to utilization-based tests.
\newblock In {\em Real-Time Systems Symposium (RTSS)}, 2015.

\bibitem{conf:/rtns09/Davis}
R.~Davis, T.~Rothvoß, S.~Baruah, and A.~Burns.
\newblock Quantifying the sub-optimality of uniprocessor fixed priority
  pre-emptive scheduling for sporadic tasksets with arbitrary deadlines.
\newblock In {\em Real-Time and Network Systems (RTNS)}, pages 23--31, 2009.

\bibitem{DBLP:journals/csur/DavisB11}
R.~I. Davis and A.~Burns.
\newblock A survey of hard real-time scheduling for multiprocessor systems.
\newblock {\em ACM Comput. Surv.}, 43(4):35, 2011.

\bibitem{DBLP:journals/rts/DavisRBB09}
R.~I. Davis, T.~Rothvo{\ss}, S.~K. Baruah, and A.~Burns.
\newblock Exact quantification of the sub-optimality of uniprocessor fixed
  priority pre-emptive scheduling.
\newblock {\em Real-Time Systems}, 43(3):211--258, 2009.

\bibitem{DBLP:conf/rtss/EisenbrandR08}
F.~Eisenbrand and T.~Rothvo\ss.
\newblock Static-priority real-time scheduling: Response time computation is
  {NP-Hard}.
\newblock In {\em Proceedings of the 29th {IEEE} Real-Time Systems Symposium,
  {RTSS} 2008, Barcelona, Spain, 30 November - 3 December 2008}, pages
  397--406, 2008.

\bibitem{DBLP:conf/ecrts/FisherB05}
N.~Fisher and S.~K. Baruah.
\newblock A fully polynomial-time approximation scheme for feasibility analysis
  in static-priority systems with arbitrary relative deadlines.
\newblock In {\em ECRTS}, pages 117--126, 2005.

\bibitem{DBLP:conf/ecrts/FisherBB06}
N.~Fisher, S.~K. Baruah, and T.~P. Baker.
\newblock The partitioned scheduling of sporadic tasks according to
  static-priorities.
\newblock In {\em 18th Euromicro Conference on Real-Time Systems, ECRTS'06, 5-7
  July 2006, Dresden, Germany, Proceedings}, pages 118--127, 2006.

\bibitem{DBLP:journals/siamam/Graham69}
R.~L. Graham.
\newblock Bounds on multiprocessing timing anomalies.
\newblock {\em SIAM Journal of Applied Mathematics}, 17(2):416--429, 1969.

\bibitem{DBLP:journals/jacm/HochbaumS87}
D.~S. Hochbaum and D.~B. Shmoys.
\newblock Using dual approximation algorithms for scheduling problems
  theoretical and practical results.
\newblock {\em J. ACM}, 34(1):144--162, 1987.

\bibitem{DBLP:conf/waoa/KarrenbauerR10}
A.~Karrenbauer and T.~Rothvo{\ss}.
\newblock A 3/2-approximation algorithm for rate-monotonic multiprocessor
  scheduling of implicit-deadline tasks.
\newblock In {\em Workshop of Approximation and Online Algorithms {WAOA}},
  pages 166--177, 2010.

\bibitem{journals/tc/LeeSP04}
C.-G. Lee, L.~Sha, and A.~Peddi.
\newblock Enhanced utilization bounds for qos management.
\newblock {\em IEEE Trans. Computers}, 53(2):187--200, 2004.

\bibitem{DBLP:conf/rtss/Lehoczky90}
J.~P. Lehoczky.
\newblock Fixed priority scheduling of periodic task sets with arbitrary
  deadlines.
\newblock In {\em RTSS}, pages 201--209, 1990.

\bibitem{DBLP:conf/rtss/LehoczkySD89}
J.~P. Lehoczky, L.~Sha, and Y.~Ding.
\newblock The rate monotonic scheduling algorithm: Exact characterization and
  average case behavior.
\newblock In {\em IEEE Real-Time Systems Symposium}, pages 166--171, 1989.

\bibitem{journals/pe/LeungW82}
J.~Y.-T. Leung and J.~Whitehead.
\newblock On the complexity of fixed-priority scheduling of periodic, real-time
  tasks.
\newblock {\em Perform. Eval.}, 2(4):237--250, 1982.

\bibitem{liu73scheduling}
C.~L. Liu and J.~W. Layland.
\newblock Scheduling algorithms for multiprogramming in a hard-real-time
  environment.
\newblock {\em Journal of the ACM}, 20(1):46--61, 1973.

\bibitem{Mok:1983:FDP:888951}
A.~K. Mok.
\newblock Fundamental design problems of distributed systems for the
  hard-real-time environment.
\newblock Technical report, Cambridge, MA, USA, 1983.

\bibitem{DBLP:conf/ecrts/BruggenCH15}
G.~von~der Bruggen, J.-J. Chen, and W.~Huang.
\newblock Schedulability and optimization analysis for non-preemptive static
  priority scheduling based on task utilization and blocking factors.
\newblock In {\em ECRTS}, pages 90--101, 2015.

\end{thebibliography}
}
  
\end{spacing}

\normalsize

\ifbool{techreport}{
\section*{Appendix A}

\begin{proofAppendix}{Theorem~\ref{theorem:sporadic-constrained-tight1}}
  The remaining part of the proof is to show that
  $\frac{\sum_{i=1}^{N-1} dbf(\tau_i, t)}{Mt} \leq \frac{f}{2}$ for
  all $t \geq 1$ in the input instance. Since each task has $M-1$
  duplicated copies
  with the same task characteristics, we will implicitly drop the
  $M-1$ duplicated copies by considering only these $2M$ tasks in ${\bf T}_m$. We
  index these tasks according to their periods in a non-decreasing
  order. Due to the setting of the task periods, we know that $T_i =
  f+\frac{i-1}{M-1}$ for $i=1,2,\ldots, M$ and $T_i = \infty$ for
  $i=M+1, M+2, \ldots, 2M$.

  For the rest of the proof, we will only analyze the demand bound
  function of these $2M$ tasks on one
  processor, and our objective is to show that $\frac{\sum_{i=1}^{2M}
    dbf(\tau_i, t)}{t} \leq \frac{f}{2}$ for all $t \geq 1$.  Note
  that this is identical to the original $N-1$ tasks on $M$ processors. Moreover,
  for the rest of the proof, we only consider the arrival pattern that
  defines the demand bound function of the system from time $0$ on
  the processor. 

  For $t=1$, we know that $\frac{\sum_{i=1}^{2M} dbf(\tau_i, 1)}{1}
  =\frac{f}{2}$.  For $t > 1$, we need to identify the demand
  of the $M$ periodic tasks, whose periods are not $\infty$, more precisely.  Here, we define some
  terms for the simplicity of the explanations. We define a
  \emph{pile} of jobs as follows: \emph{The $\ell$-th job of task
    $\tau_i$ is placed in the $\ell$-th pile.}  Therefore, by the
  definition, each pile has $M$ jobs and has total execution time
  equal to $1-f+\frac{1-f}{M-1} \approx 1-f$ since $M$ is
  sufficiently large. Precisely, the absolute deadline and the arrival
  time of the $j$-th job (from the earliest arrival) in the $\ell$-th
  pile are $\ell\cdot ( f+ j-1)$ and $(\ell-1)\cdot ( f+ j-1)$,
  respectively.  Therefore, the contribution of the jobs in the
  $\ell$-th pile to the demand bound function at time $t$ is \emph{at
    most}
  \begin{equation}
    \label{eq:dbf-ind-sharp}
    dbf_\ell^{\sharp}(t) = 
    \begin{cases}
      0 & \mbox{ if } t < \ell\cdot    f\\
      (t-\ell \cdot f)\cdot \frac{1}{\ell}& \mbox{ if }\ell\cdot f
      \leq t < \ell \\
      1-f & \mbox{ if } t \geq \ell
    \end{cases}
  \end{equation}

  We can now define a safe upper bound $dbf^{\sharp}(t)$ of the demand
  bound function of the $2M$ tasks for $t \geq 1$ as follows:
  \begin{equation}
    \label{eq:dbf-total-sharp}
    dbf^{\sharp}(t) =1.5f-1+ \sum_{\ell=1}^{\infty} dbf^{\sharp}_\ell(t).
  \end{equation}
  Based on such an over-approximation, $dbf^{\sharp}(t)$ is a
  piece-wise linear function, which is differentiable.
  Figure~\ref{fig:segment-curve} illustrates the above definition of
  $dbf^{\sharp}(t)$ and   $\frac{dbf^{\sharp}(t)}{t}$.

  We now prove that $\max_{t \geq 1} \frac{dbf^{\sharp}(t)}{t}$
  happens when $t$ is a postive integer. For any positive number
  $\ell$, the function $dbf^{\sharp}(t)$ for $t$ in the interval
  $[\ell, \ell+1)$ can have either two segments $\left[\ell,
    (\floor{\frac{\ell}{f}}+1)f\right),
  \left[(\floor{\frac{\ell}{f}}+1)f, \ell+1\right)$ or three segments
  $\left[\ell, (\floor{\frac{\ell}{f}}+1)f\right),
  \left[(\floor{\frac{\ell}{f}}+1)f,
    (\floor{\frac{\ell}{f}}+2)f\right)$,
  $\left[(\floor{\frac{\ell}{f}}+2)f, \ell+1\right)$ of linearly
  continuous functions.  For example, in
  Figure~\ref{fig:segment-curve-A}, in interval $[2, 3)$, there are
  three segments and, in interval $[3, 4)$, there are only two
  segments. 

  The case when $\ell$ is $1$ is clear, since $dbf^{\sharp}(t) =
  1.5f-1+1-f=0.5f$ when $1 \leq t< 2f$ and $dbf^{\sharp}(t) = 0.5f +
  (t-2f)/2$ when $2f \leq t < 2$. Therefore, $dbf^{\sharp}(t)/t
  \leq dbf^{\sharp}(1)$ when $1 \leq t < 2$. Similarly,
  $dbf^{\sharp}(t)/t \leq dbf^{\sharp}(3)/3$ when $2 \leq t < 3$,
  as also illustrated in Figure~\ref{fig:segment-curve-B}.

  For a given $t \geq 3$, by the definition that $f \approx
  0.7034674$, we also know that $\floor{\frac{t}{f}} \geq \floor{t} + 1$.
  Consider the case that there are two linear segments of the function
  $dbf^{\sharp}(t)$ in $[\ell, \ell+1)$ for an integer $\ell \geq
  3$. In this case, these two segments are $\left[\ell,
    (\floor{\frac{\ell}{f}}+1)f\right),
  \left[(\floor{\frac{\ell}{f}}+1)f, \ell+1\right)$. When $\ell \leq t
  \leq (\floor{\frac{\ell}{f}}+1)f$, we know that
  \begin{align}
  dbf^{\sharp}(t) &=1.5f-1+\ell(1-f) + \sum_{i=\ell+1}^{\floor{\frac{t}{f}}} (t-i \cdot f)\cdot \frac{1}{i}\nonumber\\    
  &= dbf^{\sharp}(\ell)+ (t-\ell)\cdot\sum_{i=\ell+1}^{\floor{\frac{\ell}{f}}}  \frac{1}{i}. \label{eq:dbf-sharp-1st-segment}
  \end{align}
 When $(\floor{\frac{\ell}{f}}+1)f \leq t < \ell+1$, we know that 
{\small  \begin{align}
  dbf^{\sharp}(t)   &= dbf^{\sharp}(\ell)+ (t-\ell)\cdot\left(\sum_{i=\ell+1}^{\floor{\frac{\ell}{f}}}  \frac{1}{i}\right) + \frac{t-(\floor{\frac{\ell}{f}}+1)f}{\floor{\frac{\ell}{f}}+1}. \nonumber\\
&= dbf^{\sharp}\left((\floor{\frac{\ell}{f}}+1)f\right)+ \left(t-(\floor{\frac{\ell}{f}}+1)f\right)\cdot\sum_{i=\ell+1}^{\floor{\frac{\ell}{f}}+1}  \frac{1}{i}. \label{eq:dbf-sharp-2nd-segment}
  \end{align}}

For a linear segment started from $t^*$, we have
$\frac{dbf^{\sharp}(t^*+x)}{t^*+x} = \frac{dbf^{\sharp}(t^*) + \sigma
  x}{t^*+x}$ if $x$ is no more than the length of the linear segment, where
$\sigma$ is the slope of the linear function defined in
Eq.~\eqref{eq:dbf-sharp-1st-segment} or
Eq.~\eqref{eq:dbf-sharp-2nd-segment}.  Since the first order
derivative of $\frac{dbf^{\sharp}(t^*) + \sigma x}{t^*+x}$ with
respect to $x$ is $\frac{\sigma t^* - dbf^{\sharp}(t^*)}{(t^*+x)^2}$,
we know that the function $\frac{dbf^{\sharp}(t^*+x)}{t^*+x}$
monotonically decreases, or monotonically increases, or remains the
same with respect to valid $x$ values. 
If, for contradiction, there exists a $t$ with
   $\ell \leq t < \ell+1$ such that $\frac{dbf^{\sharp}(t)}{t} >
   \frac{dbf^{\sharp}(\ell)}{\ell}$ and $\frac{dbf^{\sharp}(t)}{t} >
   \frac{dbf^{\sharp}(\ell+1)}{\ell+1}$, then $\frac{dbf^{\sharp}(t)}{t}$ in the first segment must be
   increasing and $\frac{dbf^{\sharp}(t)}{t}$  in the second segment must be decreasing.
However, 
 with
Eq.~\eqref{eq:dbf-sharp-1st-segment} and
Eq.~\eqref{eq:dbf-sharp-2nd-segment}, it is rather clear that the
\emph{slope} of $dbf^{\sharp}(t)$ of the second linear segment is
larger than the slope of the first linear segment when there are two
linear segments in $[\ell, \ell+1)$ and $\ell \geq 3$.  Therefore,
if $\frac{dbf^{\sharp}(t)}{t}$ is an increasing function in the first segment, 
it must be also an increasing function in the second segment. Hence,
$dbf^{\sharp}(t)/t \leq \max\left\{\frac{dbf^{\sharp}(\ell)}{\ell},
  \frac{dbf^{\sharp}(\ell+1)}{\ell+1}\right\}$ when $\ell \leq t <
\ell+1$ for the case when there are two linear segments of
$dbf^{\sharp}(t)$ in $[\ell, \ell+1)$.
  The
  same argument holds for the case with three segments as well by
  examining the \emph{increasing slopes} of $dbf^{\sharp}(t)$ of these three segments. As a
  result, we can conclude that $\max_{t \geq 1}
  \frac{dbf^{\sharp}(t)}{t}$ happens when $t$ is a postive integer.

  The remaining part of the proof is to show that the maximum
  $\frac{dbf^{\sharp}(t)}{t}$ happens when $t=1$ or $t=\infty$, in
  which both lead $\frac{dbf^{\sharp}(t)}{t}$ to
  $\frac{f}{2}$. We can easily evaluate $\frac{dbf^{\sharp}(t)}{t}$
  for $t=1,2,3,4,5$ as follows:
  \begin{itemize}
  \item $\frac{dbf^{\sharp}(1)}{1}=1.5f-1+1-f = 0.5f \approx 0.3517337$.
  \item $\frac{dbf^{\sharp}(2)}{2}=\frac{1.5f-1+2-2f}{2} \approx 0.3241$.
  \item $\frac{dbf^{\sharp}(3)}{3}=\frac{1.5f-1+3-3f+(3-4f)/4}{3}
    \approx 0.3304$.
  \item $\frac{dbf^{\sharp}(4)}{4}=\frac{1.5f-1+4-4f+(4-5f)/5}{4}
    \approx 0.3344$.
  \item $\frac{dbf^{\sharp}(5)}{5}=\frac{1.5f-1+5-5f+(5-6f)/6+(5-7f)/7}{5}
    \approx 0.3357$.
  \end{itemize}
  For any positive integer $\ell$ with $\ell \geq 5$, we can
  reformulate Eq.~\eqref{eq:dbf-total-sharp} into the following
  equation:
  \begin{align}
    dbf^{\sharp}(\ell) =&1.5f-1+ \sum_{i=1}^{\ell} (1-f) +
    \sum_{i=\ell+1}^{\floor{\ell/f}} \frac{\ell-i\cdot
      f}{i}\nonumber\\
    =&1.5f-1+ \ell-\floor{\ell/f}f + \sum_{i=\ell+1}^{\floor{\ell/f}}
    \frac{\ell}{i}
    \label{eq:density-total-sharp}
  \end{align}

  We can complete the proof by showing that
  $\frac{dbf^{\sharp}(\ell+1)}{\ell+1} -
  \frac{dbf^{\sharp}(\ell)}{\ell} \geq 0$ for any positive integer
  $\ell \geq 5$. Such a fact can be seen in Figure~\ref{fig:dbf-tight}
  based on numerical evaluations, but a formal proof requires quite
  some effort due to the floor function.  For any positive integer $\ell$ with $\ell
  \geq 5$, we have
{\footnotesize
\begin{align}
&  \frac{dbf^{\sharp}(\ell+1)}{\ell+1} -
  \frac{dbf^{\sharp}(\ell)}{\ell}\nonumber\\
=&
  \frac{1.5f-1+\ell+1-\floor{(\ell+1)/f}f}{\ell+1}
   -
  \frac{1.5f-1+\ell-\floor{\ell/f}f}{\ell}\nonumber\\
& \; - \frac{\ell \frac{1}{\ell+1}}{\ell}+ \sum_{i=\floor{\ell/f}+1}^{\floor{(\ell+1)/f}} \frac{(\ell+1)
  \frac{1}{i}}{\ell+1}\nonumber\\
=& 
  -\frac{1.5f-1-\floor{\ell/f}f}{\ell(\ell+1)} -
  \frac{f(\floor{(\ell+1)/f}-\floor{\ell/f})}{\ell+1}\nonumber\\
& \; - \frac{1}{\ell+1} + \sum_{i=\floor{\ell/f}+1}^{\floor{(\ell+1)/f}} 
  \frac{1}{i}\label{eq:density-diff}
  \end{align}
} Since $\ell \geq 5$, the above equation is well-defined. Due to the
fact that $f \approx 0.7034674$, we know that $\floor{(\ell+1)/f}$ is
either $\floor{\ell/f}+1$ or $\floor{\ell/f}+2$. Let $\frac{\ell}{f} =
\floor{\frac{\ell}{f}}+b$ where $0 \leq b < 1$. If $0 \leq b <
2-\frac{1}{f}$, then $\frac{\ell+1}{f} <
\floor{\frac{\ell}{f}}+2-\frac{1}{f}+\frac{1}{f} =
\floor{\frac{\ell}{f}}+2$, which implies that
$\floor{(\ell+1)/f}=\floor{\ell/f}+1$ for such a case. If $2-\frac{1}{f} \leq b <
1$, then $\floor{\frac{\ell}{f}}+1+\frac{1}{f} > \frac{\ell+1}{f} \geq
\floor{\frac{\ell}{f}}+2-\frac{1}{f}+\frac{1}{f} =
\floor{\frac{\ell}{f}}+2$, which implies that
$\floor{(\ell+1)/f}=\floor{\ell/f}+2$ for such a case.
We now analyze these two cases individually.

{\bf Case 1: $0 \leq b < 2-\frac{1}{f}$: } In this case,
$\floor{(\ell+1)/f}=\floor{\ell/f}+1$. Therefore, Eq.~\eqref{eq:density-diff}
becomes
\begin{align}
  \label{eq:Y-ell-b-case1}
Y_1(b)=\frac{(\frac{\ell}{f}-b)f-1.5f+1}{\ell(\ell+1)} - \frac{f+1}{\ell+1} +
\frac{1}{\frac{\ell}{f}-b+1}.  
\end{align}
We take the first order derivative of $Y_1(b)$ with respect to $b$. Since
$\frac{d Y_1(b)}{d b} = \frac{-f}{\ell(\ell+1)} +
\frac{1}{(\frac{\ell}{f}-b+1)^2}$ is an increasing function with
respect to $b$ and $\frac{-f}{\ell(\ell+1)} +
\frac{1}{(\frac{\ell+1}{f}-1)^2} =
f\left(\frac{-(\ell+1)^2+2f(\ell+1)-f^2+f(\ell+1)^2-f(\ell+1)}{\ell(\ell+1)(\ell+1-f)^2}\right)<0$,
for a given positive $\ell$,
the function $Y_1(b)$ is at least $Y_1(2-\frac{1}{f})$. Since
\begin{align*}
  Y_1(2-\frac{1}{f}) 
= &\frac{\ell-3.5f+2}{\ell(\ell+1)} - \frac{f+1}{\ell+1} +
\frac{1}{\frac{\ell+1}{f}-1}\\
=& \frac{-f(\ell+1) -2.5f+2}{\ell(\ell+1)} +\frac{f}{\ell+1-f}\\
=& \frac{(\ell+1)(f^2+2-3.5f) + 2.5f^2-2f}{\ell(\ell+1)(\ell+1-f)}\\
> & 0, \qquad\mbox{[\bf due to $\ell \geq 5$.]}
\end{align*}
we know that $Y_1(b) > 0$ for any integer $\ell \geq 5$ and $0 \leq b
< 2-\frac{1}{f}$. Therefore, $\frac{dbf^{\sharp}(\ell+1)}{\ell+1} -
  \frac{dbf^{\sharp}(\ell)}{\ell} > 0$ for such a case.

{\bf Case 2: $2-\frac{1}{f} \leq b < 1$: } In this case,
$\floor{(\ell+1)/f}=\floor{\ell/f}+2$. Therefore,
Eq.~\eqref{eq:density-diff} becomes
\begin{align}
  \label{eq:Y-ell-b-case2}
Y_2(b)=\frac{\ell -bf -1.5f+1}{\ell(\ell+1)} - \frac{2f+1}{\ell+1} +
\frac{1}{\frac{\ell}{f}-b+1} +
\frac{1}{\frac{\ell}{f}-b+2}.  
\end{align}
We take the first order derivative of $Y_2(b)$ with respect to $b$. Since $\frac{d Y_2(b)}{d b} = \frac{-f}{\ell(\ell+1)} +
\frac{1}{(\frac{\ell}{f}-b+1)^2} + \frac{1}{(\frac{\ell}{f}-b+2)^2}$ is an increasing function with respect to $b$, and $\frac{-f}{\ell(\ell+1)} +
\frac{1}{(\frac{\ell+1}{f} -1)^2} + \frac{1}{(\frac{\ell+1}{f})^2} >
\frac{-f}{\ell(\ell+1)} + \frac{2}{(\frac{\ell+1}{f})^2} =
f\left(\frac{-(\ell+1)+2f\ell}{\ell(\ell+1)^2}\right)> 0$, for a
given $\ell \geq 3$, the function $Y_2(b)$ is
at least $Y_2(2-\frac{1}{f})$. Since
\begin{align*}
  Y_2(2-\frac{1}{f}) 
= &\frac{\ell-3.5f+2}{\ell(\ell+1)} - \frac{2f+1}{\ell+1} +
\frac{1}{\frac{\ell+1}{f}-1}+\frac{1}{\frac{\ell+1}{f}}\\
=& Y_1(2-\frac{1}{f})
>  0, \qquad\mbox{[\bf due to $\ell \geq 5$.]}
\end{align*}
we know that $Y_2(b) > 0$ for any integer $\ell \geq 5$ and
$2-\frac{1}{f} \leq b < 1$. Therefore,
$\frac{dbf^{\sharp}(\ell+1)}{\ell+1} - \frac{dbf^{\sharp}(\ell)}{\ell}
> 0$ for such a case.

With the above two cases, we can conclude that
$\frac{dbf^{\sharp}(\ell)}{\ell} \leq \frac{dbf^{\sharp}(\infty)}{\infty}
= \frac{f}{2}$ for any positive integer $\ell \geq 5$, which concludes
the proof.
\end{proofAppendix}

}{}
\end{document}